\documentclass[11pt]{article}

\def\vecY{Z_1}
\def\vecZ{Z_2}

\usepackage[absolute,overlay]{textpos}

\usepackage{nicefrac}

\usepackage{float}
\usepackage[dvips]{graphicx}
\usepackage{psfrag}

\usepackage{amsmath,amsthm,amsfonts,amssymb,amscd}
\usepackage[table]{xcolor}
\usepackage{mathrsfs}
\usepackage{upgreek}
\usepackage{stackrel}
\usepackage{tipa}
\usepackage{accents}
\usepackage[multiple]{footmisc}

\usepackage{multirow} 
\usepackage{booktabs}

\numberwithin{equation}{section}

\definecolor{SmartBlue}{RGB}{51, 51, 255}

\usepackage{enumitem}

\setlist[itemize]{topsep=0pt, partopsep=0pt, parsep=0pt, itemsep=0pt}

\usepackage{hyperref}

\hypersetup{filecolor=magenta}
\hypersetup{colorlinks=true}
\hypersetup{urlcolor=SmartBlue}
\hypersetup{linkcolor=SmartBlue}
\hypersetup{pdfpagemode=UseNone}
\hypersetup{citecolor=SmartBlue}

\usepackage[absolute,overlay]{textpos}

\usepackage{changepage}

\usepackage{float}

\usepackage{amsmath,amsthm,amsfonts,amssymb,amscd}
\usepackage{mathrsfs}
\usepackage{upgreek}
\usepackage{stackrel}
\usepackage{tipa}
\usepackage{accents}
\usepackage[multiple]{footmisc}

\usepackage{multirow} 
\usepackage{booktabs}

\usepackage{multicol}

\usepackage{xfrac} 

\newtheorem{theorem}{Theorem}[section]
\newtheorem{proposition}[theorem]{Proposition}
\newtheorem{lemma}[theorem]{Lemma}
\newtheorem{definition}[theorem]{Definition}

\newtheorem{remark}[theorem]{Remark}

\newtheorem{setup}[theorem]{Setup}

\newtheorem*{orient*}{Orientation Condition}
\newtheorem*{junction*}{Junction Conditions}

\newtheorem*{nonexphor*}{Non-expanding horizons}
\newtheorem*{weakisolhor*}{Weakly isolated horizons}
\newtheorem*{isolhor*}{Isolated horizons}

\def\gen{\kil}
\def\ovgen{\ovkil}

\def\ovkil{\ov{\teta}}
\def\ovkil{\ov{\teta}}
\def\teta{\eta}
\def\kil{\teta}

\def\defi{{\stackrel{\mbox{\tiny {\textbf{def}}}}{\,\, = \,\, }}}

\def\nablao{{\stackrel{\circ}{\nabla}}}
\def\Gammao{{\stackrel{\circ}{\Gamma}}}

\def\sone{\hat{s}}
\def\bsone{\bs{\hat{s}}}

\def\metdata{\{\mathcal{N},\gamma,\ellc,\elltwo\}}

\def\G{\mathfrak{G}}

\def\bkilone{\bs{\upvarpi}}

\def\n{\mathfrak{n}}

\def\G{\mathcal G}
\def\Fcal{\mathcal F}

\def\bY{\textup{\textbf{Y}}}

\def\Y{\textup{Y}}
\def\bU{\textup{\textbf{U}}}
\def\U{\textup{U}}

\def\bF{\textup{\textbf{F}}}
\def\F{\textup{F}}

\def\Yn{r}
\def\Q{Q}

\newcommand{\nn}{\nonumber}

\newcommand{\hate}{\hat{e}}
\newcommand{\bm}[1]{\mbox{\boldmath $#1$}}
\newcommand\ovnabla{\ov{\nabla}}

\newcommand\N{\mathcal N}
\newcommand\M{\mathcal M}

\newcommand\elltwo{\ell^{(2)}}

\newcommand\hypdata{\{ \mathcal{N},\gamma,\ellc, \elltwo, \bY\}}
\newcommand\rig{\zeta}

\newcommand\A{\mathcal A}

\def\bmell{\bm{\ell}}
\def\ellc{\bmell}

\def\nablao{{\stackrel{\circ}{\nabla}}}
\def\K{\mathcal K}
\def\Kgen{\K}
\def\Kkil{\K}

\def\Sigmakil{\Sigma}

\def\defi{{\stackrel{\mbox{\tiny \textup{\textbf{def}}}}{\,\, = \,\, }}}

\def\nablao{{\stackrel{\circ}{\nabla}}}
\def\Gammao{{\stackrel{\circ}{\Gamma}}}

\def\nabh{\nabla^h}

\def\sone{s}
\def\bsone{\bs{s}}

\def\metdata{\{\mathcal{N},\gamma,\ellc,\elltwo\}}

\def\n{\mathfrak{n}}

\def\G{\mathcal G}
\def\Fcal{\mathcal F}

\def\bY{\textup{\textbf{Y}}}

\def\Y{\textup{Y}}

\def\bU{\textup{\textbf{U}}}
\def\U{\textup{U}}

\def\bF{\textup{\textbf{F}}}
\def\F{\textup{F}}

\def\Yn{r}
\def\Q{\kappa_n}

\def\bmell{\bm{\ell}}
\def\ellc{\bmell}

\def\nablao{{\stackrel{\circ}{\nabla}}}
\def\K{\mathcal K}

\newcommand{\ov}{\overline}

\newcommand{\cp}{\partial}

\newcommand{\bs}{\boldsymbol}
\newcommand{\lp}{\left(}
\newcommand{\rp}{\right)}

\newcommand{\cu}{\mathcal{U}}
\newcommand{\cv}{\mathcal{V}}

\newcommand{\lb}{\left\lbrace}
\newcommand{\rb}{\right\rbrace}

\newcommand{\ld}{\left.}

\newcommand{\rv}{\right\vert}

\newcommand{\tdo}{d}
\newcommand{\lieo}{\mathsterling}

\def\metdataa{\{\gamma,\ellc,\elltwo\}}
\def\hypdataa{\{\gamma,\ellc,\elltwo,\bY\}}

\newcommand{\tcb}{}

\newcommand{\nullhyp}{\widetilde{\N}}

\newcommand{\spc}{\textup{ }}

\newcommand{\ke}{\widetilde{\kappa}}

\usepackage[new]{old-arrows}

\newcommand{\lightfrak}[1]{\scalebox{0.85}[1.2]{$\mathfrak{#1}$}}

\newcommand{\w}{\lightfrak{w}}
\newcommand{\p}{\lightfrak{p}}
\newcommand{\qone}{\mathfrak{q}}
\newcommand{\bqone}{\bs{\qone}}
\def\bcalY{\bs{\calY}}
\def\calY{\mathfrak{I}}

\def\calP{\mathfrak{B}}
\def\bcalP{\bs{\mathfrak{B}}}

\makeatletter
\newsavebox\myboxA
\newsavebox\myboxB
\newlength\mylenA

\newcommand*\xoverline[2][0.75]{%
    \sbox{\myboxA}{$\m@th#2$}%
    \setbox\myboxB\null
    \ht\myboxB=\ht\myboxA%
    \dp\myboxB=\dp\myboxA%
    \wd\myboxB=#1\wd\myboxA
    \sbox\myboxB{$\m@th\overline{\copy\myboxB}$}
    \setlength\mylenA{\the\wd\myboxA}
    \addtolength\mylenA{-\the\wd\myboxB}%
    \ifdim\wd\myboxB<\wd\myboxA%
       \rlap{\hskip 0.5\mylenA\usebox\myboxB}{\usebox\myboxA}%
    \else
        \hskip -0.5\mylenA\rlap{\usebox\myboxA}{\hskip 0.5\mylenA\usebox\myboxB}%
    \fi}
\makeatother

\DeclareFontFamily{U}{rcjhbltx}{}
\DeclareFontShape{U}{rcjhbltx}{m}{n}{<->rcjhbltx}{}
\DeclareSymbolFont{hebrewletters}{U}{rcjhbltx}{m}{n}

\DeclareMathSymbol{\aleph}{\mathord}{hebrewletters}{39}
\DeclareMathSymbol{\beth}{\mathord}{hebrewletters}{98}
\DeclareMathSymbol{\gimel}{\mathord}{hebrewletters}{103}
\DeclareMathSymbol{\lamed}{\mathord}{hebrewletters}{108}
\DeclareMathSymbol{\mem}{\mathord}{hebrewletters}{109}
\DeclareMathSymbol{\ayin}{\mathord}{hebrewletters}{96}
\DeclareMathSymbol{\tsadi}{\mathord}{hebrewletters}{118}
\DeclareMathSymbol{\qof}{\mathord}{hebrewletters}{113}
\DeclareMathSymbol{\shin}{\mathord}{hebrewletters}{152}
\DeclareMathSymbol{\pe}{\mathord}{hebrewletters}{80}
\DeclareMathSymbol{\heh}{\mathord}{hebrewletters}{104}
\DeclareMathSymbol{\peh}{\mathord}{hebrewletters}{112}
 \newcounter{mnotecount}

 \newcommand{\mnote}[1]
 {\protect{\stepcounter{mnotecount}}$^{\mbox{\tiny
 $\,\bullet$\themnotecount}}$ \marginpar{
 \raggedright\tiny\em
 $\,\bullet$\themnotecount: #1} }

\allowdisplaybreaks

\usepackage[a4paper,top=25.4mm,bottom=25.4mm,
left=23mm, right=23mm,bindingoffset=0mm]{geometry} 

\title{Null hypersurface data and ambient vector fields:\\ Killing horizons of order zero/one}

\author{Miguel Manzano\footnote{miguelmanzano06@usal.es}\hspace{0.17cm} and Marc Mars\footnote{marc@usal.es}\\ \\
Instituto de F\'{\i}sica Fundamental y Matem\'aticas, IUFFyM\\
Universidad de Salamanca
}

\makeatletter
\newcommand\subsubsubsection{\@startsection{paragraph}{4}{\z@}{-2.5ex\@plus -1ex \@minus -.25ex}{1.25ex \@plus .25ex}{\normalfont\normalsize\bfseries}}
\newcommand\subsubsubsubsection{\@startsection{subparagraph}{5}{\z@}{-2.5ex\@plus -1ex \@minus -.25ex}{1.25ex \@plus .25ex}{\normalfont\normalsize\bfseries}}
\makeatother

\usepackage{tabularx}

\begin{document}

\setlength{\abovedisplayskip}{0.15cm}
\setlength{\belowdisplayskip}{0.15cm}

\maketitle

\begin{abstract}
In this work, we study null hypersurfaces admitting a privileged vector field $\gen$ which is null and tangent at the hypersurface.\ 
We derive 
an identity 
that relates 
the deformation tensor of $\gen$ with 
tensor fields codifying the intrinsic and extrinsic geometry of the hypersurface.\ This is done without imposing restrictions either on the topology of the hypersurface or on the (possibly empty) subset of points where $\gen$ vanishes.\ 
We introduce a generalized notion of surface gravity $\kappa$ that extends smoothly the usual one to the fixed points.\ 
We also analyze the properties of the Lie derivative of the Levi-Civita connection along $\gen$.\ This 
analysis 
allows us to introduce three 
new notions of 
abstractly defined horizons (i.e.\ not necessarily viewed as embedded submanifolds) that we then compare with the standard concepts of non-expanding, weakly isolated and isolated horizons.\ The former generalize the latter to completely general topologies and to horizons admitting fixed points.\ Finally, we study the structure of the fixed points on these new horizons. 
\end{abstract}

%
%

\section{Introduction}\label{c:Introduction}


Many spacetimes of physical and/or geometrical relevance admit privileged vector fields.\ A primary example are spacetimes endowed with Killing fields, but there are many more interesting possibilities.\  
Besides the natural generalization of a spacetime admitting a less restrictive type of symmetry such as a homothety or a conformal Killing vector, it can also happen that there is one observer (modelled, as usual, by a unit future timelike  
vector) that is physically or geometrically privileged, e.g.\ by being geodesic, shear-free,  irrotational or any combination thereof.\ Privileged null vector fields are also commonplace, for instance when a spacetime is algebraically special so that the Weyl tensor admits a multiple principle null direction, or when the spacetime admits a Kerr-Schild decomposition.\ The examples are endless.

This paper is devoted to studying the properties of \textit{null} 
hypersurfaces embedded in spacetimes
admitting 
a privileged vector field which is \textit{null} and \textit{tangent} at the hypersurface.\ Our main aim is to understand the effect of one such 
vector 
on 
the geometry of the hypersurface.\ We denote the hypersurface by $\nullhyp$ and the vector field by $\kil$.\  

Our approach to understand the interplay between $\gen$ and the hypersurface $\nullhyp$ 
is based upon obtaining a set of 
identities that relate the geometric properties of $\gen$ with the intrinsic and extrinsic geometry of $\nullhyp$.\ 
These identities will have the property that \textit{all} the dependence on $\gen$ and on the hypersurface will be fully explicit.\  
In this way, 
whenever some additional information about $\gen$ is known (for instance if $\gen$ happens to be a Killing vector, an homothety or it is restricted in any manner), the identities will become restrictions on the geometry of the hypersurface.\  
To make our results 
applicable to a variety of situations, 
we impose no restriction on the topology of the hypersurface nor on the nature of the zero set $\mathcal{S}\subset\nullhyp$ of $\gen$ (namely the set of points of the hypersurface where $\gen$ vanishes). 

The identities we derive have many potential applications.\ 
In particular, in this paper they allow us to introduce three interesting notions of horizons, namely \textit{abstract Killing horizon of order zero} (AKH$_0$, Definition \ref{defKHD0}), \textit{abstract weak Killing horizon of order one} (AWKH$_1$, Definition \ref{def:weak:KH1}), and \textit{abstract Killing horizon of order one} (AKH$_1$, Definition \ref{defKH1}).\ As we will see, these definitions 
can be formulated at a purely non-embedded level, namely 
without the necessity of assuming 
the existence of an ambient spacetime.\  
Moreover, they generalize the important and well-known notions of non-expanding, weakly isolated and isolated horizons \cite{ ashtekar2000generic, ashtekar2002geometry, ashtekar2000isolated, gourgoulhon20063+, jaramillo2009isolated,krishnan2002isolated}.\ The generalization is in three directions.\ Firstly, \textit{we separate very clearly the concepts that can be defined viewing the hypersurfaces as fully detached from the spacetime}.\ Secondly, in our notions \textit{the generator of the horizon is allowed to vanish 
at the horizon}.\ Thirdly, we place \textit{no restrictions whatsoever upon the topology of the horizon}.\  
%
%
These generalized notions of Killing horizons are interesting for several reasons.\ For instance, near horizon geometries are defined in terms of Killing horizons with vanishing surface gravity.\ By definition, Killing horizons exclude the fixed point set of the Killing vector.\ This can have unpleasant consequences such as transforming cross-sections that would otherwise be compact into non-compact ones.\ By being able to incorporate the zero set of the generators into the framework, this type of difficulties simply disappear.\ It is remarkable that already at the level of abstract Killing horizons of order one the structure of the zero set of the generator is 
analogous as for 
Killing horizons (cf.\ Theorem \ref{lemzerosets}).\ 

Further applications of the results presented here will be analyzed in \cite{manzano2023master} (see also \cite{manzano2023PhD}) and in \cite{mars2024transverse}.\ In \cite{manzano2023master}, we obtain an equation --- valid for completely general null hypersurfaces equipped with a null and tangent privileged vector field --- that generalizes the well-known near horizon equation of isolated horizons \cite{ashtekar2002geometry} and the master equation of multiple Killing horizons \cite{mars2018multiple}.\ In \cite{mars2024transverse} the authors prove uniqueness of spacetimes that solve the $\Lambda$-vacuum field equations to infinite order from data on a Killing horizon. 

A crucial aspect to explore the interaction between $\gen$ and 
$\nullhyp$ is to be able to codify their geometric properties.\ Our main tool to describe the hypersurface is the \textit{formalism of hypersurface data} 
\cite{manzano2023PhD, manzano2023constraint, manzano2023matching, mars2013constraint, mars2020hypersurface, mars2024abstract,  mars2023covariant,  Mars2023first,mars2024transverse}, a mathematical framework for studying general hypersurfaces in a 
detached way from any ambient space.\ This formalism allows one to fully encode the intrinsic geometry of a hypersurface 
in a $(0,2)$-tensor $\gamma$, a one-form $\ellc$ and a scalar $\elltwo$ subject to specific conditions.\ In particular, from $\metdataa$ one can reconstruct the full ambient metric $g$ at the hypersurface whenever it happens to be embedded.\ The extrinsic geometry of the  hypersurface, on the other hand, can be entirely captured by an additional $(0,2)$-tensor $\bY$ which encodes first transverse derivatives of $g$ evaluated at the hypersurface.\ 
Regarding $\gen$, its geometric properties can be codified in a variety of ways.\ In this paper, we will 
use its deformation tensor $\Kgen\defi\pounds_{\gen}g$.\ 
However, we will see that in some cases it is more useful to codify the geometric information not directly with the Lie derivative of the metric, but with the Lie derivative of the Levi-Civita connection along $\kil$.\ This derivative defines a $(1,2)$-tensor 
which we denote by $\Sigmakil$.\  
This tensor captures all the information of the first derivatives of the deformation tensor $\Kkil$ of $\kil$ as well as curvature aspects of the ambient space.\

The main results of this paper are the following.\ We demonstrate that, even at the detached level, one can introduce a notion of a surface gravity $\kappa$ associated with $\kil$.\ Remarkably, $\kappa$ is defined not only at points where $\kil\neq0$, but everywhere on $\nullhyp$. Furthermore, at points where $\kil\neq0$,  
it coincides with the standard notion of surface gravity of a null, tangent vector field (see \eqref{nabla:eta:eta:1st}).\ In fact, $\kappa$ extends the standard surface gravity \textit{smoothly} to all points of $\nullhyp$ where $\kil$ vanishes.\ 
We also derive an identity for the tangent components of the first transverse Lie derivative of $\Kgen$ at the hypersurface in terms of $\gen\vert_{\nullhyp}$, $\Kgen\vert_{\nullhyp}$ and $\hypdataa$.\ 
In 
the case when 
the 
information about $\kil$
is codified
in terms of $\Sigmakil$ instead of $\Kkil$ we obtain, 
among other things, 
an explicit expression for $\Sigmakil(X,W)$,  
where $X,W$ are tangent to $\nullhyp$ (Lemma \ref{lem:Sigma:eta:general}).\ This result turns out to be the essential 
in the construction of the notions of AKH$_0$, AWKH$_1$ and AKH$_1$.\  
%
%
As stressed above, the relevance of these identities involving $\Kkil$ and $\Sigmakil$ 
relies on the fact that, when additional information about $\gen$ is known, they 
automatically restrict the geometry of the hypersurface in a practical and useful way.


The structure of the paper is as follows.\ In Section \ref{sec:prelim:data} we revisit the basic concepts and results concerning the formalism of hypersurface data 
in the null case.\ 
%
%
The generalized notion of surface gravity $\kappa$ is presented in Section \ref{secNullCasep4}.\ 
In Section \ref{secLieYp4} we obtain the identity relating the deformation tensor $\Kgen$ of $\gen$ and the tensors $\hypdataa$.\  
Section \ref{app:rig(sigma)} focuses on the properties of the tensor $\Sigmakil$.\ The notions of Killing horizons of order zero and one are presented in Section \ref{sec:HD-KH0and1} and then compared with the standard concepts of non-expanding, weakly isolated and isolated horizons in Section \ref{sec:connection:NEH,MKH}. The paper concludes with Section \ref{sec:application:zeroes}, where we analyze the structure of the set of fixed points of a generator $\kil$ of a Killing horizon of order one.

\subsection{Notation and conventions}\label{sec:notation}

In this paper, all manifolds are smooth, connected and without boundary.\ Given a manifold $\M$ we use $\mathcal{F}\lp\mathcal{M}\rp\defi C^{\infty}\lp\mathcal{M},\mathbb{R}\rp$ and $\mathcal{F}^{\star}\lp\mathcal{M}\rp\subset\mathcal{F}\lp\mathcal{M}\rp$ for the subset of no-where zero functions.\ The tangent bundle is denoted 
by $T\mathcal{M}$ and $\Gamma\lp T\mathcal{M}\rp$ is the set
of sections (i.e.\ vector fields).\ We use $\pounds$, $d$ for the Lie derivative and exterior derivative.\ Both tensorial and abstract index notation will be employed.\  
We work in arbitrary dimension $\mathfrak{n}$ and use the following sets of indices:
\begin{equation}
\label{notation}
\alpha,\beta,...=0,1,2,...,\mathfrak{n};\qquad a,b,...=1,2,...,\mathfrak{n};\qquad A,B,...=2,...,\mathfrak{n}.
\end{equation}
When index-free notation is used (and only then) we shall distinguish
covariant tensors with boldface.\ As usual, parenthesis (resp.\ brackets) denote symmetrization (resp.\ antisymmetrization) of indices.\ The symmetrized tensor product is defined by $A\otimes_s B\equiv\frac{1}{2}(A\otimes B+B\otimes A)$.\  
In any spacetime $(\mathcal{M},g)$,
the scalar product of two vectors is written 
as $g(X,Y)$,  
and we use $g^{\sharp}$, $\nabla$  for the inverse metric and Levi-Civita derivative of $g$ respectively.\  
%
%
Our signature for Lorentzian manifolds $\lp \mathcal{M},g\rp$ is $(-,+, ... ,+)$.

\section{Preliminaries: formalism of hypersurface data}\label{sec:prelim:data}

In this section we summarize the results of the  formalism of hypersurface data needed for this paper.\ 
Since this work focuses on null hypersurfaces, we only introduce the formalism in the null case.\ We refer the reader to \cite{mars2013constraint, mars2020hypersurface} for a more thorough introduction to the formalism. 

\textit{Null metric hypersurface data} $\metdata$ is an $\n$-dimensional manifold
$\mathcal{N}$ endowed with a symmetric $(0,2)$-tensor  $\gamma $ 
of signature $(0,+,...,+)$, 
a covector
$\ellc$, and a scalar $\ell^{(2)}$ subject to the condition that the
symmetric $(0,2)$-tensor $\bs{\mathcal{A}}\vert_p$ on $T_p\mathcal{N}\times\mathbb{R}$ given by 
\begin{equation}
\nn
\begin{array}{c}
  \ld\bs{\mathcal{A}}\rv_p
  \lp\lp W,a\rp,\lp Z,b\rp\rp\defi \ld \gamma \rv_p\lp W,Z\rp+a\ld\ellc\rv_p\lp Z\rp+b\ld\ellc\rv_p\lp W\rp+ab \ell^{(2)}\vert_p,\quad W,Z\in T_p\mathcal{N},\quad a,b\in\mathbb{R}
\end{array}
\end{equation}
is non-degenerate at every $p\in\mathcal{N}$.\  \textit{Null hypersurface data} $\hypdata$ is null metric hypersurface data with an additional symmetric $(0,2)$-tensor $\bY$.\ The tensor $\bs{\mathcal{A}}\vert_p$ is non-degenerate, so it admits a unique ``inverse'' $(2,0)$-tensor  $\mathcal{A}\vert_p$.\ Its splitting
\begin{equation}
\label{ambientinversemetric}
\ld\mathcal{A}\rv_p\lp\lp \bs{\alpha},a\rp,\lp \bs{\beta},b\rp\rp\defi \ld P\rv_p\lp \bs{\alpha},\bs{\beta}\rp+a\ld n\rv_p\lp \bs{\beta}\rp+b\ld n\rv_p\lp \bs{\alpha}\rp
%
%
,\quad
\bs{\alpha},\bs{\beta}\in T^{\star}_p\mathcal{N},\quad a,b\in\mathbb{R}
\end{equation}
defines a symmetric $(2,0)$-tensor $P$ and a vector $n$ 
in $\N$ satisfying
\cite{mars2013constraint}:

\vspace{-0.65cm}

\noindent
\begin{minipage}[t]{0.175\textwidth}
\begin{align}
		\gamma_{ab} n^b & = 0, \label{prod1} 
\end{align}
\end{minipage}
\hfill
\hfill
\begin{minipage}[t]{0.17\textwidth}
	\begin{align}
		\ell_a n^a & = 1, \label{prod2}  
	\end{align}
\end{minipage}
\hfill
\begin{minipage}[t]{0.3\textwidth}
	\begin{align}
		P^{ab} \ell_b + \elltwo n^a & = 0,  \label{prod3} 
	\end{align}
\end{minipage}
\hfill
\begin{minipage}[t]{0.3\textwidth}
	\begin{align}
		P^{ab} \gamma_{bc} + n^a \ell_c & = \delta^a_c. \label{prod4}
	\end{align}
\end{minipage}

\vspace{-0.cm}

The vector field $n$ is no-where zero (by \eqref{prod2}) and defines the degenerate direction of $\gamma$ (by \eqref{prod1}).\ Thus, 
the radical of $\gamma$ at any 
$p\in\N$, i.e.\ the set $\textup{Rad}\gamma\vert_p \defi \{X\in T_p \N \spc\vert\spc\gamma(X,\cdot)=0\}$, is spanned by $n\vert_p$. 
We call the integral curves of $n$  \textit{generators} of $\N$.


It is useful to introduce the following tensor fields
\cite{mars2020hypersurface}: 

\vspace{-0.6cm}

\noindent
\begin{minipage}[t]{0.55\textwidth}
	\begin{align}
		\label{threetensors} \bF & \defi \frac{1}{2} d \ellc, & \bm{\sone} &\defi  \bF(n,\cdot), & \bU  &\defi  \frac{1}{2}\pounds_{n} \gamma ,
	\end{align}
\end{minipage}
\hfill
\hfill
\begin{minipage}[t]{0.44\textwidth}
\begin{align}
	\label{defY(n,.)andQ}
	\bs{\Yn}&\defi\bY(n,\cdot),& \Q &\defi -\bY(n,n).
\end{align}
\end{minipage}

\vspace{-0.cm}

Observe that $\bU$ is symmetric and $\bF$ is a $2$-form, hence 
$\bsone(n)=0$.\ The tensors $\bsone$ and $\bU$ satisfy 
\cite{mars2020hypersurface}  

\vspace{-0.7cm}

\noindent
\begin{minipage}[t]{0.5\textwidth}
	\begin{align}
		\pounds_{n} \ellc & = 2 \bm{\sone} , \label{poundsellc}
	\end{align}
\end{minipage}
\hfill
\hfill
\begin{minipage}[t]{0.45\textwidth}
\begin{align}
\bU (n, \cdot ) & = 0.
\label{Un} 
\end{align}
\end{minipage}

\vspace{-0.cm}

Any null metric hypersurface data $\metdata$ admits a torsion-free connection $\nablao$,
called \textit{metric hypersurface connection}, uniquely defined by the equations
\cite{mars2020hypersurface} 

\vspace{-0.25cm}

\begin{multicols}{2}
\noindent
\begin{align}
\nablao_{a} \gamma_{bc}  =& - \ell_b \U_{ac} - \ell_c \U_{ab}, \label{nablaogamma} \\
\nablao_a \ell_b  =&\spc \F_{ab} - \elltwo \U_{ab}. \label{nablaoll} 
\end{align}
\end{multicols}

\vspace{-0.25cm}

For later use, we include two $\nablao$-derivatives of $n$ \cite{mars2020hypersurface} together with 
an identity for the Lie derivative of $\gamma$ along a vector $\hat{\bs{\sigma}}\defi P(\bs{\sigma},\cdot)$,  $\bs{\sigma}\in\Gamma(T^{\star}\N)$ 
\cite[Lemma 2.2]{manzano2023matching}:  

\vspace{-0.6cm}

\noindent
\begin{minipage}[t]{0.34\textwidth}
	\begin{align}
		\label{nablaonnull}
		\begin{array}{l}
			\nablao_bn^c=n^c\sone_b+P^{ac}\U_{ab},
		\end{array}
	\end{align}
\end{minipage}
\hfill
\hfill
\begin{minipage}[t]{0.15\textwidth}
	\begin{align}
		\label{nablao_nn}
		\hspace{-0.7cm}\nablao_nn =0,
	\end{align}
\end{minipage}
\hfill
\hfill
\begin{minipage}[t]{0.41\textwidth}
\begin{align}
     \hspace{-0.43cm}\frac{1}{2} {\pounds}_{\hat{\bs{\sigma}}} \gamma_{ab}   &= \nablao_{(a} \bs{\sigma}_{b)}
      -  \ell_{(a} \nablao_{b)} (\bs{\sigma}(n)). 
      \label{deromegagamma} 
\end{align}
\end{minipage}

\vspace{-0cm}

The hypersurface data formalism and the actual geometry of hypersurfaces in spacetimes are linked through the 
notion of embeddedness of the data 
\cite{mars2020hypersurface}.\ 
A null metric data $\metdata$ is said to be embedded in an $(\mathfrak{n}+1)$-spacetime $\lp \mathcal{M},g\rp$
if there exists an embedding $\phi :\mathcal{N}\longhookrightarrow\mathcal{M}$ and a rigging $\zeta$ (i.e.\ a vector field along $\phi \lp\mathcal{N}\rp$, everywhere transverse to it) satisfying 
\begin{equation}
\label{emhd}
\phi ^{\star}\lp g\rp= \gamma , \qquad\phi ^{\star}\lp g\lp\zeta,\cdot\rp\rp=\ellc, \qquad\phi ^{\star}\lp g\lp\zeta,\zeta\rp\rp=\elltwo.
\end{equation}
For null hypersurface data $\hypdata$ one requires, in addition, that
\begin{equation}
\label{YtensorEmbDef}
\dfrac{1}{2}\phi^{\star}\lp \pounds_{\zeta}g\rp=\bY.
\end{equation}
Since the definition of embedded data involves both $\phi$ and
$\rig$ we shall say ``embedded with embedding $\phi$ and rigging $\rig$", or
``$\{\phi,\rig\}$-embedded" for short.\ When clear from the context we will identify scalars and vectors on $\N$ with their counterparts on $\phi(\N)$.\ 
We will use the word \textit{abstract} to refer to mathematical objects defined in terms of the hypersurface data only, i.e.\ that do not require the data to be embedded in an ambient space (e.g.\ $\gamma$, $\bU$ and even the manifold $\N$).

In the embedded case, 
the 
connection 
$\nablao$ and the Levi-Civita derivative $\nabla$ of $g$ are related by \cite{mars2020hypersurface}
\begin{align}
	\label{nablaXYnablao}\nabla_{X}W&= \nablao_{X}W-\bY(X,W)
	n - \bU(X,W)\rig,  
	\quad\forall X,W\in\Gamma(T\N).
\end{align}
The choice of a rigging vector is highly non-unique.\ The hypersurface data formalism captures this fact through a built-in gauge freedom \cite{mars2020hypersurface}.\ 
Given null hypersurface data $\hypdata$, a non-zero function $z\in\mathcal{F}^{\star}\lp\mathcal{N}\rp$ and a vector field $V\in\Gamma\lp T\mathcal{N}\rp$, the gauge transformed data 
\begin{align*}
\G_{(z,V)}\Big(\hypdata\Big)\defi \lb \mathcal{N},\mathcal{G}_{\lp z,V\rp}\lp \gamma \rp,\mathcal{G}_{\lp z,V\rp}\lp\ellc\rp,\mathcal{G}_{\lp z,V\rp}\big( \ell^{(2)} \big),\mathcal{G}_{\lp z,V\rp}( \bY )\rb 
\end{align*}
is given by 
\begin{align}
	\label{gaugegamma&ell2} \hspace{-0.2cm}\mathcal{G}_{\lp z,V\rp}(\gamma)& \defi  \gamma ,\hspace{0.42cm}\mathcal{G}_{\lp z,V\rp}\lp\ellc\rp \defi z\lp\ellc+ \gamma \lp V,\cdot\rp\rp,\hspace{0.42cm}\mathcal{G}_{\lp z,V\rp}\big(\ell^{(2)}\big) \defi z^2\big( \ell^{(2)}+2\ellc\lp V\rp+ \gamma \lp V,V\rp\big),\\
	\label{gaugeY}\hspace{-0.2cm}\mathcal{G}_{\lp z,V\rp}(\bY)  & \defi z\bY+ \ellc\otimes_s \tdo z+\frac{1}{2}\lieo_{zV} \gamma. 
\end{align} 
The set of transformations $\{\G_{(z,V)}\}$ forms a group \cite{mars2020hypersurface}
and has the property that,    
%
if $\hypdata$ is  $\{\phi,\rig\}$-embedded in $(\mathcal{M},g)$, then the transformed data $\mathcal{G}_{(z,V)}(\hypdata)$ is 
embedded in the same space with the same embedding and rigging $\mathcal{G}_{(z,V)} (\rig)  \defi  z (\rig + \phi _{\star} V)$.\ The gauge group induces transformations on \textit{all} other geometric objects, e.g.\ $n$, $P$, $\nablao$.\ In this paper we shall need the  following two \cite{manzano2023constraint, mars2020hypersurface}:
%

\vspace{-0.6cm}

%
%
%
%
%
\noindent
\begin{minipage}[t]{0.5\textwidth}
\begin{align}
	\label{gaugen}\mathcal{G}_{\lp z,V\rp}\lp n\rp&=z^{-1}n ,\qquad\qquad\quad\spc
\end{align}
\end{minipage}%
\begin{minipage}[t]{0.49\textwidth}
\begin{align}
	\mathcal{G}_{\lp z,V\rp}\lp \Q \rp & = z^{-2}\big( \Q z - n(z)  \big).\label{Qprime}
\end{align}
\end{minipage}

\vspace{-0cm}

We conclude the preliminaries by introducing a setup so that we can later refer to it easily.
%
%
\begin{setup}\label{setup:basis:e_a}
	For null hypersurface data $\hypdata$
	$\{\phi,\rig\}$-embedded in a spacetime $(\mathcal{M}, g)$, we select any local basis $\{ \hate_a \}$ of $\Gamma(T\mathcal{N})$ and define $\{e_a \defi \phi _{\star} \hate_a\}$.\ Then,  
	$\{\zeta,e_a\}$ is a 
	basis of $\Gamma(T\mathcal{M})\vert_{\phi (\mathcal{N})}$.\ The hypersurface $\phi (\mathcal{N})$ admits a unique null normal covector $\bs{\nu}$ satisfying $\bs{\nu}(\zeta) = 1$.\ By construction, $\bs{\nu}$ is an element of the dual basis of $\{\zeta,e_a\}$, which we denote by $\{\bs{\nu}, \bs{\theta}^a\}$.\ We define the vector fields 
	$\nu \defi g^{\sharp}(\bs{\nu},\cdot)$, ${\theta}^a \defi g^{\sharp}(\bs{\theta}^a,\cdot)$ and the covector 
	$\bm{\rig} \defi  g(\rig,\cdot)$ along $\phi(\N)$.
\end{setup}
By definition of dual basis (and by \eqref{prod1}-\eqref{prod4}), $\{\nu,\theta^a\}$ decompose in the basis $\{\rig, e_a\}$ as 

\vspace{-0.25cm}

\begin{multicols}{2}
	\noindent
	\begin{align}
		\label{normal}
		\nu & = n^a e_a, \\
		\theta^a & = P^{ab} e_b + n^a \rig. \label{omegas}
	\end{align}
\end{multicols}

\vspace{-0.4cm}

When a metric data $\metdata$ is embedded, the components of $\bs{\A}$ in the basis $\{ (\hat{e}_a, 0), (0,1)\}$ coincide (by \eqref{emhd}) with those of $g$ in the basis $\{ e_a, \rig\}$. This, together with \eqref{ambientinversemetric}, means that 
\begin{align}
	\label{gup} g^{\mu\nu} & \stackbin{\phi (\mathcal{N})}{=} 
	n^c \left ( \rig^{\mu} e_c^{\nu} +
	\rig^{\nu}e_c^{\mu}  \right ) +P^{cd} e_c^{\mu} e_d^{\nu}
	\qquad \Longleftrightarrow\qquad g^{\mu\nu}\stackbin{\phi (\mathcal{N})}{=} e_c^{\mu} \theta^c{}^{\nu} + \rig^{\mu} \nu^\nu,
\end{align}
where the equivalence follows from \eqref{normal}-\eqref{omegas}.\ 
Observe that $g (\nu, \nu ) = g^{\sharp}(\bs{\nu}, \bs{\nu} ) = 
0$, so 
\eqref{gup} is consistent with 
$\nu$ being null at the hypersurface $\phi(\N)$.\ As mentioned before, the integral curves of $n$ are called null generators.\ This name is justified by \eqref{normal}, as $\nu$ is itself a null generator of $\phi(\N)$.\ 
%
%
It is also worth mentioning that,  
\textit{in the null case, 
the tensor $\bU$ 
coincides \textup{\cite{mars2020hypersurface}} 
with the second fundamental form of the hypersurface $\phi(\N)$ with respect to $\nu$.}

For later use, we include the tangent $\nabla$-derivative of the rigging $\rig$  
\cite{mars2013constraint}: 
\begin{align}
	\nabla_{e_a} \rig =\frac{1}{2} \big(\nablao_a \elltwo\big) \nu
	+ \left ( \Y_{ab} + \F_{ab} \right ) \theta^b.
	\label{nablarig3}
\end{align}

\section{Abstract notion of surface gravity}\label{secNullCasep4}

In later sections we shall introduce and study
Killing horizons of order zero and one in the framework of hypersurface data.\ Much in the same way as standard Killing horizons, they will also involve a privileged vector field which is null and tangent to the hypersurface.\ This vector field is a property of the hypersurface itself, so it cannot depend on any choice of transverse direction to the hypersurface or, at the abstract level, on any choice of gauge.\ 
In preparation for that situation, we now consider general null hypersurface data $\hypdata$ endowed
with an extra \textit{gauge-invariant} vector field $\ovkil$ \textit{in the radical of $\gamma$}.\ 
In this section we shall prove that to any such $\ovkil$ one can associate a surface gravity on $\N$.\ This definition is interesting in two respects.\ First, it encodes at the abstract level the usual notion of surface gravity of (embedded) null hypersurfaces. Second, and particularly interesting, 
it 
is well-defined everywhere on $\N$, including the points where $\ovkil$ vanishes.

Since $\ovkil$ lies in the radical of $\gamma$, it is necessarily proportional to $n$.\ We 
let $\alpha\in\Fcal(\N)$ be defined by $\ovkil=\alpha n$, and 
$\mathcal{S}$ be the (possibly empty) set of zeroes of $\ovkil$,  
i.e.\ $\mathcal{S}\defi \{p\in\N\spc\vert\spc \alpha(p)=0\}$.\ 
The following lemma introduces the notion of surface gravity of $\ovkil$ and proves its gauge-invariance. 
\begin{lemma}\label{lemkappagauge} Let  $\hypdata$  be
null hypersurface data endowed with gauge-invariant vector field 
$\ovkil=\alpha n$. Then, the function 
\begin{equation}
	\label{defkappaonN}\kappa\defi d\alpha\lp n \rp-\alpha\bY\lp n,n\rp
\end{equation}
is gauge-invariant.
\end{lemma}
\begin{proof}
By hypothesis $\ovkil=\mathcal{G}_{(z,V)}(\ovkil)$ for any gauge group element.
By \eqref{gaugen}
this implies
\begin{equation}
	\mathcal{G}_{(z,V)}(\alpha)=z\alpha. \label{gaugef}
\end{equation}
To check that $\kappa$ is gauge-invariant we start with the second term (recall \eqref{gaugeY}):
\begin{align}
	\mathcal{G}_{(z,V)}&\big(\alpha\bY\lp n,n\rp\big)=z\alpha\Big( z\bY+ \ellc\otimes_s \tdo z+\frac{1}{2}\lieo_{zV} \gamma \Big)\big( z^{-1}n,z^{-1}n\big) 
	= \alpha\bY(n,n)+z^{-1}\alpha d z(n),
	\label{gaugefYnn} 
\end{align}
where 
we used $\ellc(n)=1$ 
(cf.\ \eqref{prod2}) 
and $(\pounds_{zV} \gamma) (n,n)=\pounds_{zV} (\gamma(n,n))-2\gamma(\pounds_{zV}n,n)=0$. Now, the term $d\alpha(n)$ transforms as
\begin{equation}
	\label{gaugedfn}\mathcal{G}_{(z,V)}(d\alpha(n))=d\big(\mathcal{G}_{(z,V)}\alpha\big)\big( z^{-1}n\big)= z^{-1}( \alpha dz+zd\alpha)(n)=z^{-1} \alpha dz(n)+d\alpha(n).
\end{equation}
Combining this with \eqref{gaugefYnn} proves 
that 
$\G_{(z,V)}(\kappa)=\kappa$ indeed.
\end{proof}
%
%
When 
$\hypdata$ is $\{\phi,\rig\}$-embedded in $(\M,g)$,  
the 
vector field $\eta\defi \phi_{\star}\ovkil$ is null and tangent to 
$\phi(\N)$ at all its points.\ Hence,  away from the zeroes of $\kil\vert_{\phi(\N)}$ (i.e.\ on 
$ \phi (\N \setminus\mathcal{S})$), 
the standard notion of surface gravity $\widetilde{\kappa}$ of $\kil$ can be introduced by means of 
\begin{align}
\label{nabla:eta:eta:1st}\nabla_{\teta}\teta=\widetilde{\kappa}\teta. 
\end{align}
It turns out that 
the pull-back of $\widetilde{\kappa}$ to $\N\setminus\mathcal{S}$ is precisely  $\kappa$, as we show next.\  
The interesting fact is that definition \eqref{defkappaonN} does not require $\alpha$ to be non-zero.\ By construction $\kappa$ is smooth and well-defined \textit{everywhere on $\N$}.\ It is not obvious a priori that the spacetime function $\widetilde{\kappa}$, which in general is defined only in  a subset of $\phi(\N)$, extends smoothly to all the hypersurface.\ This is actually an interesting consequence of the following computation that uses $\bU(\ovkil,\cdot)=0$ (cf.\ \eqref{Un}), \eqref{nablao_nn}, \eqref{nablaXYnablao} 
and the definitions of $\ke$ and $\kappa$:  
\begin{align}
	\nonumber 
	\widetilde{\kappa}\eta&=\nabla_{\eta}\eta=\phi_{\star}\big(\nablao_{\ovkil}\ovkil-\bY(\ovkil,\ovkil)n\big)=\alpha \phi_{\star}\lp \alpha \nablao_{n}n+\lp d\alpha (n)-\alpha \bY(n,n)\rp n\rp
	=\kappa\hspace{0.035cm}\kil\quad\text{on}\quad\phi(\N\setminus\mathcal{S}).
\end{align} 
%
%
Thus, $\widetilde{\kappa}\circ\phi=\kappa$ on $\N\setminus\mathcal{S}$, which 
justifies calling the \textit{abstract} function $\kappa$ surface gravity.\ 


\section{Lie derivative of $\bY$ along a tangent null vector field}
\label{secLieYp4}

The purpose of this section is to study the interplay between 
the geometry of embedded null hypersurfaces 
and the existence  
of a privileged vector field in the ambient space.\ More specifically, we shall derive identities that relate 
$\{\phi,\rig\}$-embedded null hypersurface data $\hypdata$ with the deformation tensor of 
a vector field $\gen$, defined on a neighbourhood of $\phi(\N)$, which becomes null and tangent to $\phi(\N)$ at all its points.\ The idea is that  when $\gen$ is special in the sense that we have information about its deformation tensor, then one automatically obtains restrictions on the hypersurface data.\ 
The results 
in this section have many potential consequences and have already found application in several circumstances (see \cite{manzano2023PhD,manzano2023master,mars2024transverse}). One of them, on which we will focus in the paper, is the construction of abstract notions of Killing horizons of order zero and one.\ 

In fact, 
the analysis of this section can be performed for an \textit{arbitrary} privileged vector field (non-necessarily tangent or null at the hypersurface) as well as for \textit{general} (i.e.\ non-necessarily null) hypersurface data.\ Such analysis can be found in the Ph.D.\ thesis of the first named author  \cite[Chap.\ 5]{manzano2023PhD}.\ 
For the sake of conciseness, here we only present the results in the case of null data and a null tangent privileged vector field $\kil$.\ 

Consider null hypersurface data $\hypdata$ $\{\phi,\rig\}$-embedded in a spacetime $(\M,g)$.\ Suppose that $(\M,g)$ admits a privileged vector field $\gen$ in a neighbourhood of $\phi(\N)$ with the properties of being tangent to $\phi(\N)$ and null on this hypersurface.\ 
We assume Setup \ref{setup:basis:e_a} and 
use 
$\ovkil$ to denote the counterpart of $\kil$ on $\N$, i.e.\ $\kil\defi \phi_{\star}\ovkil$.\ By construction $\ovkil$ is  gauge-invariant and belongs to the radical of $\gamma$, so there exists a function  $\alpha\in\Fcal(\N)$ such that $\ovkil \defi \alpha n$.\ This means, in particular, that all results in Section \ref{secNullCasep4} apply.\ 
The so-called deformation tensor of $\gen$ is defined by
\begin{equation}
\label{defDeforTensor}\Kgen \defi \pounds_{\gen} g.
\end{equation}
%
%
%
It is useful to encode information of this tensor and its first transverse derivative at the hypersurface by means of the following two functions $\w,\p$, covector $\bqone$, and symmetric $(0,2)$-tensor $\bcalY$ on $\N$: 

\vspace{-0.5cm}


\begin{minipage}[t]{0.719\textwidth}
	\begin{align}
		\hspace{-0.3cm}\label{defwpqb}\p&\defi \phi^{\star}(\Kkil \left (\rig,  \rig  \right )),\quad \w\defi \phi^{\star}(\Kkil \left (\rig,  \nu  \right )),\quad  \bqone\defi \phi^{\star}(\Kkil \left (\rig, \cdot \right )),
	\end{align}
\end{minipage}%
\hfill
\hfill
\begin{minipage}[t]{0.281\textwidth}
	\begin{align}
		\label{defwpqb:Y}\bcalY&\defi \dfrac{1}{2}\phi ^{\star}\left ( \pounds_{\rig} \Kkil \right ).
	\end{align}
\end{minipage}

\vspace{0.1cm}

%
%
%
%
%
%
%
%
%
%
%
%
%
%
Note that \eqref{defwpqb} involve only transverse components of the deformation tensor $\Kkil$.\ There is no need to introduce a name for the tangential components because they are codified by $\bU$.\ Indeed,
\begin{equation}
	\label{pullbackdeftensornull}
	\phi ^{\star} \Kkil   =
	\phi^{\star}\big(\pounds_{\kil}g\big)=
	 \pounds_{\ovkil} \gamma=
	\pounds_{\alpha n} \gamma=
	2\alpha\bU,
\end{equation}
where we have made use of the identities 

\vspace{-0.7cm}

\noindent
\begin{minipage}[t]{0.5\textwidth}
\begin{align}
	\phi ^{\star}\lp \pounds_{\phi _{\star}X} T\rp = \pounds_{X} (\phi ^{\star} T),
	\label{LiePull}
\end{align}
\end{minipage}
\hfill
\hfill
\begin{minipage}[t]{0.49\textwidth}
\begin{align}
	\label{liefngammaANDU}
	\pounds_{\varrho n}\gamma&=2\varrho\bU,
\end{align}
\end{minipage}

\vspace{-0cm}

valid for any covariant tensor field $T$ on $\M$, any vector field $X\in \Gamma(T\N)$ and any function $\varrho\in\Fcal(\N)$ (in particular, \eqref{liefngammaANDU} follows from combining \eqref{prod1} 
and $\pounds_{n}\gamma=2\bU$, cf.\ \eqref{threetensors}).\ Observe that $\bqone$ and $\w$ verify 
$\bqone(n)=\w$. 

The following theorem finds
an identity for the Lie derivative of the tensor field $\bY$ along $\ovkil$ in terms of the deformation tensor of $\kil$ at $\phi(\N)$, the pull-back of its first transverse derivative and hypersurface data.\  
Their relevance, as mentioned before, relies on the fact that they establish a connection between the ambient properties of the vector $\kil$ and the geometry of the hypersurface.
\begin{theorem}\label{propLieetaY}
	Consider null hypersurface data $\hypdata$ $\{\phi,\rig\}$-embedded in a spacetime $(\M,g)$ 
	and assume  
	Setup \ref{setup:basis:e_a}.\ Let $\gen$ be a vector field 
	in a neighbourhood of $\phi(\N)$, null and tangent 
	at $\phi(\N)$.\ 
	Define 
	$\ovkil$ by $		\gen\vert_{\phi (\mathcal{N})} \defi \phi _{\star} \ovgen$, 
	$\alpha\in\Fcal(\N)$ by  $\ovkil=\alpha n$, and 
	$A_{\gen} \in \Fcal(\N)$,  
	$X_{\gen} \in \Gamma(T\mathcal{N})$  by

\vspace{-0.65cm}


\begin{minipage}[t]{0.37\textwidth}
\begin{align}
	A_{\gen} & \defi 
	n^a (\pounds_{\ovgen} \ellc)_a - \w \label{Aeta}, 
\end{align}
\end{minipage}
\hfill
\hfill
\begin{minipage}[t]{0.62\textwidth}
	\begin{align}
	X_{\gen} & \defi \left ( 
	\frac{1}{2}\big(  \ovgen(\elltwo)-\p \big) n^a + P^{ab}\big(  (\pounds_{\ovgen} \ellc)_b-\qone_b\big) \right ) \hate_a. \label{Xeta}
\end{align}
\end{minipage}

\vspace{-0.1cm}

Then, the commutator $[\kil,\rig]$ is given by the two equivalent expressions	
\begin{align}
	[\gen, \rig] \stackbin{\phi(\N)}= A_{\gen} \rig + \phi _{\star} X_{\gen}
	\qquad\text{and}\qquad 
	[\gen, \rig] \stackbin{\phi (\mathcal{N})}{=}\frac{1}{2} \left ( \ovgen(\elltwo)
	-  \p \right ) \nu
	+ \Big( (\pounds_{\ovgen} \ellc)_a
	- \qone_a  
	\Big) \theta^a. \label{commutator}
\end{align}
	Moreover, the derivative $\pounds_{\ovgen} \bY$ satisfies the identity 
	\begin{align}
	\pounds_{\ovgen} \bY \spc \stackbin{\N}=&\spc A_{\gen} \bY + \ellc \otimes_s  \tdo A_{\gen}+ \frac{1}{2} \pounds_{X_{\gen}} \gamma+ \bcalY,
	\label{poundsY}
	\end{align}	
which in fully expanded form reads
\begin{align}
	\nn 
	\pounds_{\ovkil} \Y_{bd} \spc \stackbin{\N}=&\spc \nablao_{b}\nablao_{d}\alpha+2\alpha\nablao_{(b}\sone_{d)}+2\sone_{(b}(\nablao_{d)}\alpha)+n(\alpha) \Y_{bd}  +\dfrac{1}{2} \ovkil(\elltwo)\U_{bd} \\
	\label{final2}  & \spc -\Big(\nablao_{(b}\qone_{d)}+\w \Y_{bd}+\dfrac{\p}{2}\U_{bd}-\calY_{bd}\Big).
\end{align}	
\end{theorem}
\begin{proof}
Throughout the proof, we identify $\N$ and $\phi(\N)$ as well as scalars, vectors and one-forms on $\N$ with their counterparts on $\phi(\N)$.\ All equalities are meant to hold in the hypersurface.\ 
We start by finding explicit expressions for 
the derivatives $\nabla_{\kil}\rig$ and $\nabla_{\rig}\kil$.\ For $\nabla_{\kil}\rig$, we combine \eqref{nablaoll}, \eqref{nablarig3} and the fact that $\bU(\ovkil,\cdot)=0$ (cf.\ \eqref{Un})
to get
\begin{align}
	\label{keyforlater} \nabla_{\gen} \rig & =  \ovgen^b \nabla_{e_b} \rig= \frac{1}{2} \ovgen (\elltwo) \nu+ \ovgen^b \big( \Y_{ba} + \F_{ba} \big) \theta^a 	=  \frac{1}{2} \ovgen (\elltwo) \nu+ \ovgen^b \big(  \nablao_{b}\ell_a+\Y_{ba}   \big) \theta^a. 
\end{align}
The calculation of $\nabla_{\rig} \gen$, based on the expression $g^{\mu\nu}= e_c^{\mu} \theta^c{}^{\nu} + \rig^{\mu} \nu^\nu$ (cf.\  \eqref{gup}), is as follows: 
\begin{align}
	\rig^{\mu}\nabla_{\mu} \gen^{\beta} & = g^{\alpha\beta} \rig^{\mu} 
	\nabla_{\mu} \gen_{\alpha}  \nonumber 
	=\left ( e_a^{\alpha} \theta^a{}^{\beta} + \rig^{\alpha} \nu^{\beta}\right ) \rig^{\mu}\nabla_{\mu} \gen_{\alpha} \nonumber 
	= \theta^a{}^{\beta}e_{a}^{\alpha} \rig^{\mu} \nabla_{\mu} \gen_{\alpha}  +
	\frac{1}{2} \rig^{\alpha} \rig^{\mu} \left ( \nabla_{\mu} \gen_{\alpha}+ \nabla_{\alpha} \gen_{\mu} \right )  \nu^{\beta} \nonumber \\
	& = \theta^a{}^{\beta}\big( \qone_a  - g( \rig, \nabla_{e_a} \gen)\big) + \frac{\p}{2} \nu^{\beta}, \label{inter}
\end{align}
where in the last step we used $\Kgen_{\alpha\beta} = 2\nabla_{(\alpha} \gen_{\beta)}$ twice.\ We elaborate the second term using 
\eqref{nablarig3}:
\begin{align}
	g( \rig, \nabla_{e_a} \gen  ) & =\nabla_{e_a}\big( g (\rig, \gen )\big) - g( \nabla_{e_a} \rig, \gen)
	=  \nablao_a (  \ell_b \ovgen^b )  - \left ( \Y_{ab} + \F_{ab} \right)\ovgen^b  = \ell_b \nablao_a \ovgen^b - \Y_{ab} \ovgen^b, \label{extra:eq:2146351} 
\end{align}
where in the last step we inserted \eqref{nablaoll} (recall $\bU(\ovkil,\cdot)=0$).\  
Substituting \eqref{extra:eq:2146351} into
\eqref{inter} gives 
\begin{align}
	&\nabla_{\rig} \gen = \spc 
	\frac{\p}{2}  \nu+\Big( \qone_a- \ell_b \nablao_a \ovgen^b + \Y_{ab} \ovgen^b  \Big)  \theta^a.
	\label{Iden2Prop} 
\end{align}
The right part of \eqref{commutator} follows from inserting \eqref{keyforlater} and \eqref{Iden2Prop} into $[\kil,\rig]=\nabla_{\kil}\rig-\nabla_{\rig}\kil$ and using that $(\pounds_{\ovgen} \ellc)_{a}=\ovgen^b\nablao_b \ell_{a}+\ell_{b}\nablao_{a}\ovgen^b$.\ The left part of \eqref{commutator} is then a direct consequence of \eqref{normal}-\eqref{omegas}, definitions \eqref{Aeta}-\eqref{Xeta} and the fact that $\bqone(n)=\w$.\ 

The identity \eqref{poundsY} 
is based on the commutation property $[\pounds_X, \pounds_W] =\pounds_{[X,W]}$.\ Applying this to the ambient metric $g$ with vectors $\gen$ and $\rig$, one obtains
\begin{align}
	\pounds_{\gen} \pounds_{\rig} g = \pounds_{[\gen,\rig]} g + \pounds_{\rig} \Kgen.\label{basic}
\end{align}
This expression requires that the rigging is  extended off $\phi(\N)$, but the final result is independent of the extension, as one can easily check from \eqref{poundsY}.\ 
We now take the pull-back of \eqref{basic} to  $\mathcal{N}$.\ For the left-hand side we use identity \eqref{LiePull} and the definition \eqref{YtensorEmbDef} of $\bY$.\ For the first term in the right-hand side we simply make use of the already proven result
$[\gen, \rig] = A_{\gen} \rig + \phi _{\star} X_{\gen}$ and find
\begin{align}
	\nonumber  \phi ^{\star} ( \pounds_{[\gen,\rig]} g)&  =\phi ^{\star} \left ( \pounds_{A_{\kil}  \zeta} g + \pounds_{\phi _{\star} X_{\eta} } g \right )= \phi ^{\star} \left ( A_{\kil}  \pounds_{\zeta} g + dA_{\kil}  \otimes \bm{\zeta}+ \bm{\zeta} \otimes dA_{\kil}  + \pounds_{\phi _{\star} X_{\eta} } g \right ) \\
	\label{pullbackdeftensor} & = 2 A_{\kil}  \bY + \ellc \otimes \tdo A_{\kil} + \tdo A_{\kil}  \otimes \ellc + \pounds_{X_{\eta}} \gamma,
\end{align}
where in the last equality we used \eqref{LiePull} and \eqref{emhd}-\eqref{YtensorEmbDef}.\ Identity \eqref{poundsY} then follows at once.\ 

Finally, to obtain \eqref{final2} we first notice that, as a consequence of \eqref{poundsellc}, $\ovkil$ satisfies 
	\begin{align}
		\label{sigmathing1}
		\pounds_{\ovkil}\ellc&=\alpha\pounds_{ n}\ellc+d\alpha=2\alpha\bsone+d\alpha. 
	\end{align}	
We now define the vectors $\cv\defi P(\pounds_{\ovkil}\bs{\ell},\cdot)$ and $\mathcal{W}\defi -P(\bqone,\cdot)-\frac{\p}{2} n$ in $\N$ 
and 
prove the 
identities
	\begin{align}
		\label{lieVgammaprop1}\dfrac{1}{2}\lp\pounds_{\cv}\gamma\rp_{bd}=&\spc\nablao_{b}\nablao_{d}\alpha+2\alpha\nablao_{(b}\sone_{d)}+2\sone_{(b}\nablao_{d)}\alpha-\ell_{(b}\nablao_{d)} n(\alpha)\quad\text{and}\\
		\label{lieWgammaprop1}\dfrac{1}{2}\lp\pounds_{\mathcal{W}}\gamma\rp_{bd}=&\spc-\nablao_{(b}\qone_{d)} +\ell_{(b}\nablao_{d)}\w-\dfrac{\p}{2}\U_{bd}.
	\end{align}
	To establish \eqref{lieVgammaprop1} we particularize \eqref{deromegagamma} for $\bs{\sigma}=\pounds_{\ovkil}\bs{\ell} $ and then use \eqref{sigmathing1} to compute
	\begin{align}
		\label{middleeq25}\nablao_{(b}(\pounds_{\ovkil}\bs{\ell})_{d)}&=\nablao_{b}\nablao_{d}\alpha+2\alpha\nablao_{(b}\sone_{d)}+2\sone_{(b}\nablao_{d)}\alpha \quad\text{and}\quad (\pounds_{\ovkil}\ellc)(n)=n(\alpha).
	\end{align}  
	For \eqref{lieWgammaprop1}, we use \eqref{liefngammaANDU} and find 
	\begin{equation}
		\lp\pounds_{\mathcal{W}}\gamma\rp_{cd}=-\big(\pounds_{P(\bqone,\cdot)}\gamma\big)_{cd}-\p\U_{cd}.
	\end{equation}
	Expression \eqref{lieWgammaprop1} follows by particularizing \eqref{deromegagamma} to $\bs{\sigma}=\bqone$ and using $\bqone(n)=\w$.\  
	%
	%
	%
	%
	Once we have \eqref{lieVgammaprop1}-\eqref{middleeq25}, 
	we can compute 
	$A_{\kil}$, $\ellc\otimes_s dA_{\eta}$ and $\frac{1}{2}\pounds_{X_{\eta}}\gamma$ explicitly and thus rewrite \eqref{poundsY} as  \eqref{final2}.\
	Specifically, the second expression in \eqref{middleeq25} entails
	\begin{align}
		\label{A_Tfinal1} A_{\eta} & =n(\alpha)-\w\qquad\Longrightarrow\qquad (\ellc\otimes_s dA_{\eta})_{bd}=\ell_{(b} \Big ( \nablao_{d)}n(\alpha)-\nablao_{d)}\w \Big),
	\end{align}
	while 
	the definitions of $\cv$, $\mathcal{W}$  allow us to rewrite $X_{\kil}$ as 
	$X_{\eta} = \mathcal{W}+\frac{1}{2}  \ovkil(\elltwo) n+ \cv $ (cf.\ \eqref{Xeta}).\ Thus, 
	\begin{align}
		\label{lieX_TgammafinalANDXTfinal1} 
		\dfrac{1}{2}\pounds_{X_{\eta}}\gamma=\dfrac{1}{2}\pounds_{\mathcal{W}}\gamma+\dfrac{1}{2}\ovkil(\elltwo)\bU+\dfrac{1}{2}\pounds_{\mathcal{V}}\gamma,
	\end{align}
	where the implication is a consequence of \eqref{liefngammaANDU}.\ 
	%
	%
	Inserting \eqref{A_Tfinal1} and \eqref{lieX_TgammafinalANDXTfinal1} into \eqref{poundsY} gives \eqref{final2}
after substituting \eqref{lieVgammaprop1}-\eqref{lieWgammaprop1}.\ 
	%
\end{proof}
As mentioned before, \eqref{poundsY}-\eqref{final2} relate the deformation tensor $\Kgen$ of $\eta$ with the variation of the tensor $\bY$ along the degenerate direction defined by $\ovkil$.\ This gives a   transport equation for $\bY$ along the generators of $\N$, sourced by the ambient properties of $\eta$ codified in specific quantities constructed from $\Kgen$.\ When this tensor is known (e.g.\ if $\kil$ is a Killing vector or a less restrictive type of symmetry such as a conformal Killing vector or a homothety) then \eqref{poundsY}-\eqref{final2} provide useful information on $\pounds_{\ovkil}\bY$.\ 

As already said, identities \eqref{poundsY}-\eqref{final2} have many potential applications.\ 
In this paper we focus on its role to define and study abstract Killing horizons, but they are also 
the basis to obtain a completely general ``master equation" which extends the well-known \textit{near horizon equation} of isolated horizons (see e.g.\ \cite{ashtekar2002geometry}) as well as the so-called \textit{master equation} of multiple Killing horizons (see e.g.\ \cite{mars2018multiple}).\ These results will be presented in 
\cite{manzano2023master}.\ It is also worth mentioning that the analysis of this paper has also played a key role in \cite{mars2024transverse},  where uniqueness of spacetimes that are $\Lambda$-vacuum to infinite order is established from data on a Killing horizon.\ We emphasize that \eqref{poundsY}-\eqref{final2} are valid for completely general null hypersurfaces (with no topological restrictions) admitting a tangent null (gauge-invariant) vector.\ In particular, this vector is allowed to vanish on any subset of $\N$.

\section{The Lie derivative of the Levi-Civita connection}\label{app:rig(sigma)}

Theorem \ref{propLieetaY} involves a vector
$\kil$ defined in a neighbourhood of $\phi(\N)$.\ As already stressed, we think
of $\kil$ as a privileged vector field of which one has  information.\ So far we have encoded this information by means of its deformation tensor $\Kkil$.\ However, it may happen that the available information is at the level of
the Lie derivative of the Levi-Civita connection of $g$, which is a tensor that we denote by $\Sigmakil$, instead as in the Lie derivative of the metric.\ Obviously, the two are closely related, but studying
$\Sigmakil$ 
can give an interesting new perspective.\ In this section we study the interplay between
$\Sigmakil$ and the null hypersurface data set.\ In particular we intend to capture the essence of
$\Sigmakil$ at the abstract level.\ As we will see, 
this  
will pave the way to introducing and studying various types of abstract horizons in Section
\ref{sec:HD-KH0and1}.

The organization of the section is as follows.\ We start by defining 
$\Sigmakil$ and by summarizing some of its well-known properties.\ Then, we derive an explicit expression for the pull-back $\phi^{\star}\lp g(W,\Sigmakil)\rp$, where    
$W$ is any vector field along $\phi(\N)$, not necessarily tangent.\ 
This will allow us to find a new tensor on $\N$ which codifies information about the deformation tensor $\Kkil$ and its first transversal derivative on $\phi(\N)$. 
Finally, we provide an explicit expression for the vector field  $\Sigmakil(\phi_{\star}\vecY,\phi_{\star}\vecZ)$, where  
$\vecY,\vecZ\in\Gamma(T\N)$.

%
%


The difference of connections is a tensor, so the Lie derivative of a connection along a vector field also defines a tensor.\  
This tensor carries useful geometric information, in particular concerning the curvature of the connection.\ A standard reference for the Lie derivative of a general connection is 
\cite{yano1957Lie}.\ 
Here, as mentioned before, we shall restrict ourselves to the Lie derivative of the Levi-Civita connection $\nabla$ of $g$ along $\kil$.\ This $(1,2)$-tensor field
$\Sigmakil$ 
%
%
%
%
%
%
%
is defined by 
\begin{align}
	\Sigmakil (X,W) \defi  \pounds_{\kil}\nabla_XW-\nabla_X\pounds_{\kil}W-\nabla_{\pounds_\kil X}W, \quad
	\quad \forall X,W\in\Gamma(T\M).
	\label{defderD}
\end{align}
An alternative notation for $\Sigmakil$ that we shall also use is
$\Sigmakil = \pounds_{\kil}\nabla$.\
It is easy to prove that $\Sigmakil$ is symmetric in its covariant indices because $\nabla$ is  torsion-free.

A basic use of $\Sigmakil$ is to compute the commutator between Lie derivatives
and covariant derivatives.\ The general identity for a general $(0,p )$-tensor $T$ reads
\begin{align}
	\pounds_{\kil} \nabla_{\alpha} T_{\beta_1 \cdots \beta_p } = \nabla_{\alpha} \pounds_{\kil} T_{\beta_1 \cdots \beta_p } - \sum_{\mathfrak{i}=1}^p  \Sigmakil^{\mu}{} _{\alpha \beta_\mathfrak{i}} T_{\beta_1 \cdots \beta_{\mathfrak{i}-1} \mu \beta_{\mathfrak{i}+1} \cdots \beta_p }. \label{commutation}
\end{align}
Particularly relevant 
is the 
relation between 
$\Sigmakil$ and the Riemann tensor ${R}{}^{\mu}{}_{\beta\nu\alpha}$ of $g$, 
namely \cite{yano1957Lie}
\begin{align}
	\label{Sigma:Curvature} 
	\quad
	\Sigmakil^{\mu}{}_{\alpha\beta}&=\nabla_{\alpha} \nabla_{\beta}\kil^{\mu}+{R}{}^{\mu}{}_{\beta\nu\alpha}\kil^{\nu}.
\end{align}
We now write 
the 
relation between $\Sigmakil$ 
and 
the deformation tensor of $\eta$.\ Particularizing \eqref{commutation} for $T_{\alpha\beta}=g_{\alpha\beta}$ and performing a standard cyclic permutation of indices yields  
\begin{align}
\label{Sigma:eta:def:tensor}\Sigmakil^{\lambda}{} _{\alpha \beta} =\frac{1}{2}g^{\mu\lambda}\Big( \nabla_{\alpha} \Kkil_{\beta\mu}+\nabla_{\beta}\Kkil_{\mu\alpha} -\nabla_{\mu} \Kkil_{\alpha\beta}\Big).
\end{align}
This expression 
allows us to 
find  the pull-back $\phi^{\star}\lp g(W,\Sigmakil)\rp$ in terms of hypersurface data and the components of  $\Kkil$ and $\pounds_{\rig} \Kkil$ introduced in \eqref{defwpqb}-\eqref{defwpqb:Y}.
\begin{lemma}\label{lem:rig(sigma)}
  Let $\hypdata$ be null hypersurface data $\{\phi,\rig\}$-embedded in a spacetime $(\M,g)$.\ Consider any 
  vector field $\eta$, defined on a neighbourhood of $\phi(\N)$, which is null and tangent at  $\phi(\N)$.\ 
  Then, for any vector field $W$ along $\phi(\N)$ (not necessarily tangent), it holds 
\begin{align}
	\label{W:Sigma} \phi^{\star}\lp g(W,\Sigmakil)\rp_{ab}
	= \beta\calP_{ab}+\ov{W}{}^c\Big(  2(\nablao_{(a}\alpha) \U_{b)c}+ \U_{ab}( \qone_c-\nablao_c\alpha )+\alpha( 2\nablao_{(a}\U_{b)c}- \nablao_c\U_{ab})\Big),\\
	\calP_{ab}
	\defi  \nablao_{(a}\qone_{b)}+\w\Y_{ab}+\p\U_{ab}-\calY_{ab},\qquad\qquad\qquad\qquad\qquad\qquad\qquad\qquad\spc\spc \label{calPab:expression}.
\end{align}
where $\alpha,\beta\in\Fcal(\N)$ and  $\ov{W}\in\Gamma(T\N)$ are defined by $\eta\vert_{\phi(\N)}=\alpha\phi_{\star}n$ and
  by the decomposition
  \begin{equation}
\label{decom:W}W=\beta \rig+\phi_{\star}\ov{W}.
\end{equation}
\end{lemma}
\begin{proof}
Assume the notation in  Setup \ref{setup:basis:e_a}. From the decomposition \eqref{decom:W}, we have
\begin{align}
\label{eq:inter:sigma}\phi^{\star} (g(W,\Sigmakil))_{ab}
= \beta \rig_{\lambda}\Sigmakil^{\lambda}{} _{\alpha \beta}e^{\alpha}_ae^{\beta}_b + \ov{W}{}^c (e_c)_{\lambda}\Sigmakil^{\lambda}{} _{\alpha \beta}e^{\alpha}_ae^{\beta}_b.
\end{align}
Therefore, to prove the lemma it suffices to compute the contractions
\begin{align}
\label{Sigma:eta:cont1}(e_c)_{\lambda}\Sigmakil^{\lambda}{} _{\alpha \beta}e^{\alpha}_ae^{\beta}_b&=\frac{1}{2}e^{\alpha}_ae^{\beta}_b e_c^{\mu}\lp \nabla_{\alpha} \Kkil_{\beta\mu}+\nabla_{\beta}\Kkil_{\mu\alpha} -\nabla_{\mu} \Kkil_{\alpha\beta}\rp,\\
\label{Sigma:eta:cont2}\rig_{\lambda}\Sigmakil^{\lambda}{} _{\alpha \beta}e^{\alpha}_ae^{\beta}_b&= \frac{1}{2}e^{\alpha}_ae^{\beta}_b\rig^{\mu}\lp \nabla_{\alpha} \Kkil_{\beta\mu}+\nabla_{\beta}\Kkil_{\mu\alpha} -\nabla_{\mu} \Kkil_{\alpha\beta}\rp .
\end{align}
The first one only requires knowing $e^{\alpha}_ae^{\beta}_be_c^{\mu}\nabla_{\mu} \Kkil_{\alpha\beta} $, namely
\begin{align}
\nn e^{\alpha}_ae^{\beta}_be^{\mu}_c\nabla_{\mu} \Kkil_{\alpha\beta}=&\spc e_c^{\mu}\nabla_{\mu}(e^{\alpha}_ae^{\beta}_b \Kkil_{\alpha\beta})-\Kkil_{\alpha\beta}e^{\beta}_b e_c^{\mu}\nabla_{\mu}e^{\alpha}_a -\Kkil_{\alpha\beta}e^{\alpha}_ae_c^{\mu}\nabla_{\mu}e^{\beta}_b \\
\nn =&\spc \hate_c(2\alpha \U_{ab})-2\alpha\U_{db}\Gammao{}^{d}_{ac}+ \U_{ac}\qone_b -2\alpha\U_{da}\Gammao{}^{d}_{bc}+ \U_{bc}\qone_a \\
\nn =&\spc 2(\nablao_c\alpha) \U_{ab}+2\alpha \nablao_c\U_{ab}+ \U_{ac}\qone_b + \U_{bc}\qone_a,
\end{align}
where we have used \eqref{nablaXYnablao}, definitions \eqref{defwpqb}, \eqref{pullbackdeftensornull} and the fact that $\Kkil(\nu,\phi_{\star}X)= 2\alpha \bU(n,X)=0$ for any $X\in\Gamma(T\N)$. Replacing this in \eqref{Sigma:eta:cont1} and cancelling terms gives
\begin{align}
\label{eSigma} (e_c)_{\lambda}\Sigmakil^{\lambda}{} _{\alpha \beta}e^{\alpha}_ae^{\beta}_b= 2(\nablao_{(a}\alpha) \U_{b)c}+ \U_{ab}( \qone_c-\nablao_c\alpha )+\alpha( 2\nablao_{(a}\U_{b)c}- \nablao_c\U_{ab}) .
\end{align}
To elaborate \eqref{Sigma:eta:cont2}, we first use 
$\rig^{\mu}\nabla_{\mu} \Kkil_{\alpha\beta}=\pounds_{\rig}\Kkil_{\alpha\beta}-\Kkil_{\mu\beta}\nabla_{\alpha}\rig^{\mu}-\Kkil_{\mu\alpha}\nabla_{\beta}\rig^{\mu}$ and find 
\begin{align}
\nn \rig_{\lambda}\Sigmakil^{\lambda}{} _{\alpha \beta}e^{\alpha}_ae^{\beta}_b=&\spc 
\frac{1}{2}\Big(  e^{\alpha}_a\nabla_{\alpha} (\Kkil_{\beta\mu}e^{\beta}_b\rig^{\mu})-\Kkil_{\beta\mu}\rig^{\mu}e^{\alpha}_a\nabla_{\alpha} e^{\beta}_b\\
\nn &+e^{\beta}_b\nabla_{\beta}(\Kkil_{\mu\alpha}e^{\alpha}_a\rig^{\mu})-\Kkil_{\mu\alpha}\rig^{\mu}e^{\beta}_b\nabla_{\beta}e^{\alpha}_a -e^{\alpha}_ae^{\beta}_b\pounds_{\rig}\Kkil_{\alpha\beta}\Big).
\end{align}
Inserting now the decomposition \eqref{nablaXYnablao} and using again the definitions \eqref{defwpqb}-\eqref{defwpqb:Y} gives
\begin{align}
\nn \rig_{\lambda}\Sigmakil^{\lambda}{} _{\alpha \beta}e^{\alpha}_ae^{\beta}_b=&\spc \frac{1}{2}\Big(  \hate_a(\qone_b)-\Kkil_{\beta\mu}\rig^{\mu}(\Gammao{}^d_{ba}e_d^{\beta}-\Y_{ab}\nu^{\beta}-\U_{ab}\rig^{\beta})\\
\nn &+\hate_b(\qone_a)-\Kkil_{\mu\alpha}\rig^{\mu}(\Gammao{}^d_{ab}e_d^{\alpha}-\Y_{ab}\nu^{\alpha}-\U_{ab}\rig^{\alpha})-e^{\alpha}_ae^{\beta}_b\pounds_{\rig}\Kkil_{\alpha\beta}\Big)\\
\label{rig:Sigma} =&\spc \nablao_{(a}\qone_{b)}+\w\Y_{ab} +\p \U_{ab}-\calY_{ab}\defi \calP_{ab}.
\end{align}
Equation \eqref{W:Sigma} follows from substituting \eqref{eSigma} and \eqref{rig:Sigma} into \eqref{eq:inter:sigma}. 
\end{proof}
We have found that the symmetric $(0,2)$-tensor $\bcalP$ defined in \eqref{calPab:expression} appears naturally in the 
pull-back $\phi^{\star}\big( g(W,\Sigmakil)\big)$.\ In particular, a direct consequence of \eqref{W:Sigma} is that, when the vector $W$ is chosen to be the rigging $\rig$ (hence $\beta=1$ and $\ov{W}=0$), then \eqref{W:Sigma} simplifies to 
\begin{equation}
	\label{def:calP}\bcalP= \phi^{\star}\big( \bs{\rig}(\Sigmakil)\big),\quad \text{where}\quad\bs{\rig}\defi g(\rig,\cdot).
\end{equation}
%
%
%
Observe that \eqref{final2} can be rewritten explicitly in terms of $\bcalP$ and the derivative $\pounds_n\bY$ by inserting 
\begin{align}
	\label{lieetalien}\pounds_{\ovkil}\Y_{bd}=2r_{(b}\nablao_{d)}\alpha  +\alpha\pounds_n\Y_{bd}
\end{align}
into \eqref{final2}  and using the definition \eqref{calPab:expression} of $\bcalP$.\ Specifically, one obtains
\begin{align}
	\nn \calP_{bd} -\dfrac{\p}{2}\U_{bd}=&\spc \nablao_{b}\nablao_{d}\alpha+2(\sone_{(b}-r_{(b})\nablao_{d)}\alpha+2\alpha\nablao_{(b}\sone_{d)} \\
	\label{Lie:id:sigma} &+n(\alpha) \Y_{bd} -\alpha\pounds_n\Y_{bd} +\dfrac{\alpha}{2}  n(\elltwo)\U_{bd}.
\end{align}
As we shall see in Section \ref{sec:HD-KH0and1}, it is precisely due to  
\eqref{Lie:id:sigma} 
that 
$\bcalP$ plays a key role 
in the context of \textit{abstract} Killing horizons of order zero and one.\  
The reason 
is that, when the null hypersurface is totally geodesic (namely when the second fundamental form $\bU=0$), \eqref{Lie:id:sigma} provides an explicit expression for $\bcalP$ solely in terms of hypersurface data.\ 
Thus, when $\bU=0$ the combination $\calP_{ab}\vert_{\bU=0}
= \nablao_{(a}\qone_{b)}+\w\Y_{ab}-\calY_{ab}$ (which in principle codifies information about how $\kil$ extends off $\phi(\N)$) turns out to be  
determined by 
the geometry of the hypersurface and by the value of $\kil$ at $\phi(\N)$ exclusively.\ 
%
%
%
%

The tensor $\bcalP$ is also relevant  
in the context of 
\textit{general} null hypersurfaces admitting a privileged tangent null vector.\ In particular, in \cite{manzano2023master} (see also \cite{manzano2023PhD}) the master equation already alluded to is used to prove that $\bcalP$ governs the constancy of the surface gravity $\kappa$.\ 

To conclude the section, we find a completely general expression for the vector field  $\Sigmakil(e_a,e_b)$. 
\begin{lemma}\label{lem:Sigma:eta:general}
  Assume the hypotheses of Lemma \ref{lem:rig(sigma)} and take the notation of Setup \ref{setup:basis:e_a}. 
  Then, 
  \begin{align}
  	\nn \Sigmakil&(e_a,e_b)= \Big( \big( \w-n(\alpha) \big)\U_{ab}- \alpha  (\pounds_n \bU)_{ab} \Big)\rig\\
  	\nn &+ \Big( \nablao_{a}\nablao_{b}\alpha +2(\nablao_{(a}\alpha)(\sone_{b)}-\Yn_{b)})+ 2\alpha\nablao_{(a}\sone_{b)}+n(\alpha) \Y_{ab}-\alpha(\pounds_{n}\bY)_{ab}
  	+\frac{1}{2}\big( \alpha n(\elltwo)+ \p\big)\U_{ab}\Big)\nu\\
  	\label{Sigma:general:xprssn} &+P^{cd}\lp 2(\nablao_{(a}\alpha) \U_{b)c}	+ \U_{ab}( \qone_c-\nablao_c\alpha )+\alpha( \nablao_a\U_{bc}+ \nablao_b\U_{ca}- \nablao_c\U_{ab}) \rp e_d .
  \end{align}
In particular, if $\bU=0$ equation \eqref{Sigma:general:xprssn} becomes
  \begin{align}
	\label{Sigma:U=0:xprssn} \Sigmakil&(e_a,e_b)= \Big( \nablao_{a}\nablao_{b}\alpha +2(\nablao_{(a}\alpha)(\sone_{b)}-\Yn_{b)})+ 2\alpha\nablao_{(a}\sone_{b)}+n(\alpha) \Y_{ab}-\alpha(\pounds_{n}\bY)_{ab}\Big)\nu,
\end{align}
which is equivalent to
\begin{align}
	\Sigmakil(e_a,e_b)=
	\calP_{ab}\hspace{0.05cm} \nu.
	\label{Sigma:U=0:xprssn2} 
\end{align}
\end{lemma}
\begin{proof}
The proof relies on the fact that 
  \begin{align}
	\label{sigma(Y,Z):lemma}&\hspace{0.38cm}\Sigmakil(e_a,e_b)=\mu_{cab}n^c\rig+\big( P^{cd}\mu_{cab}+\mathfrak{r}_{ab} n^d\big) e_d,\\
	&\textup{where}\qquad  
	\label{def:v0:beta}\mathfrak{r}_{ab}\defi \phi^{\star}\Big( g\big(\rig,\Sigmakil(e_a,e_b)\big)\Big),\qquad \mu_{cab}\defi \phi^{\star}\Big( g\big(e_c,\Sigmakil(e_a,e_b)\big)\Big),
\end{align}
which follows from calculating the scalar products of \eqref{sigma(Y,Z):lemma} with the basis vectors $\{\rig,e_c\}$.\ 
It therefore suffices to compute the contraction $\mu_{cab}n^c$ and the term $P^{cd}\mu_{cab}+\mathfrak{r}_{ab} n^d$ explicitly.
Particularizing \eqref{W:Sigma} first for $W=\rig$ and then for $W=e_c$ gives

\vspace{-0.65cm}


\begin{minipage}[t]{0.22\textwidth}
	\begin{align}
		\label{def:v0} \mathfrak{r}_{ab}=&\spc \calP_{ab},
	\end{align}
\end{minipage}%
\hfill
\hfill
\begin{minipage}[t]{0.78\textwidth}
	\begin{align}
		\mu_{cab}=&\spc   2(\nablao_{(a}\alpha) \U_{b)c}
		+ \U_{ab}( \qone_c-\nablao_c\alpha )+\alpha( \nablao_a\U_{bc}+ \nablao_b\U_{ca}- \nablao_c\U_{ab}).  
		\label{def:beta(X)}
	\end{align}
\end{minipage}

\vspace{0.cm}

To obtain $\mu_{cab}n^c$ we contract  
\eqref{def:beta(X)} with $n^c$ and use $\bU(n,\cdot)=0$, $\bqone(n)=\w$ and $n^c( \nablao_a\U_{bc}+ \nablao_b\U_{ca}- \nablao_c\U_{ab})= -\U_{bc}\nablao_an^c- \U_{ca}\nablao_bn^c- n^c\nablao_c\U_{ab}=-(\pounds_n\bU)_{ab}$. This gives
 \begin{align}
	\label{def:beta(n)} \mu_{cab}n^c=&\spc \big( \w-n(\alpha) \big)\U_{ab}- \alpha  (\pounds_n \bU)_{ab}.
\end{align}
The term $P^{cd}\mu_{cab}+\mathfrak{r}_{ab} n^d$ follows directly from \eqref{def:beta(X)}, \eqref{def:v0} and \eqref{Lie:id:sigma}.\ It reads
\begin{align}
\label{P:beta:plus:v0:n} &P^{cd}\mu_{cab}+\mathfrak{r}_{ab} n^d= P^{cd}\lp 2(\nablao_{(a}\alpha) \U_{b)c}
+ \U_{ab}( \qone_c-\nablao_c\alpha )+\alpha( \nablao_a\U_{bc}+ \nablao_b\U_{ca}- \nablao_c\U_{ab}) \rp\\
\nn &+ n^d\Big( \nablao_{a}\nablao_{b}\alpha +2(\nablao_{(a}\alpha)(\sone_{b)}-\Yn_{b)})+ 2\alpha\nablao_{(a}\sone_{b)}+n(\alpha) \Y_{ab}-\alpha(\pounds_{n}\bY)_{ab}
+\frac{1}{2}\lp \alpha n(\elltwo)+ \p\rp\U_{ab}\Big).
\end{align}
Inserting \eqref{def:beta(n)} and  \eqref{P:beta:plus:v0:n} into \eqref{sigma(Y,Z):lemma} yields \eqref{Sigma:general:xprssn} (recall that $\phi_{\star}n=\nu$, cf.\ \eqref{normal}).\
Equation \eqref{Sigma:U=0:xprssn} follows from enforcing $\bU=0$ in \eqref{Sigma:general:xprssn}, while \eqref{Sigma:U=0:xprssn2} is a consequence of \eqref{Lie:id:sigma} and \eqref{Sigma:U=0:xprssn}. 
\end{proof}
Once again, the case $\bU=0$ is particularly interesting because expression 
\eqref{Sigma:U=0:xprssn} states that the
vector field $\Sigmakil(e_a,e_b)$  
can be codified entirely by the function $\alpha$ and the hypersurface data objects $n$, $\bsone$, $\nablao$ and $\bY$.\ This means, in particular, that  \textit{$\Sigmakil(e_a,e_b)$ 
is completely independent of how $\kil$ behaves off $\phi(\N)$ when the null hypersurface is totally geodesic} (i.e.\ when $\bU=0$).\ 
As mentioned before, the very same property is also satisfied by $\bcalP$ (see \eqref{Lie:id:sigma}), which is of course consistent with \eqref{Sigma:U=0:xprssn2}.\ 
As we shall see next, 
it is precisely this property that allows one to construct new abstract notions of Killing horizons of order zero and one.


\section{Abstract notions of Killing horizons of order zero and one}\label{sec:HD-KH0and1} 

Killing horizons have played a fundamental role in General Relativity, mainly because its close relation with black holes in equilibrium (see e.g.\ \cite{hawking2023large}).\ Their defining property is the existence of 
a Killing vector field $\kil$ which is null and tangent at the hypersurface.\ 
One often adds to this the condition that $\eta$ is nowhere zero at the hypersurface, but we shall not do so here.\ The deformation tensor $\Kkil$ of $\eta$ is identically zero, so in particular it vanishes together with all its derivatives at the hypersurface.\ It turns out, however, that some of the most relevant properties of Killing horizons can be fully recovered by only requiring that a few derivatives of $\Kkil$ vanish on the horizon.\ It is in these circumstances that the notions of 
Killing horizons of order zero and one arise naturally. 

By definition, a Killing horizon of order $m\in \mathbb{N} \cup \{ 0 \}$ is a null hypersurface together with a vector field $\kil$ defined in a neighbourhood of the hypersurface, such that $(i)$ $\kil$ is null and tangent at the hypersurface and $(ii)$ the transverse derivatives up to order $m$ of the deformation tensor $\Kkil$ vanish on the hypersurface.\ The purpose of this section is to provide \textit{abstract} definitions of Killing horizons of order zero and one.\ The main idea is to be able to describe such horizons \textit{in a detached way from any ambient space} where they may be embedded. 

It is to be expected that the abstract notions of Killing horizons of order zero/one rely on hypersurface data $\hypdata$ satisfying certain additional  restrictions.\ However, this already raises the question of how much geometric information from the embedded picture can be codified only in terms of $\{\gamma,\ellc,\elltwo,\bY\}$.\ For the order zero, ideally one would like the definition to enforce $\Kkil=0$ everywhere on the hypersurface.\  However, as we already know, only the pull-back of $\Kkil$ can be expressed solely in terms of the data, namely by means of $\bU$ (see \eqref{pullbackdeftensornull}).\ The remaining components of $\Kkil$ are given by $\bqone$, $\p$ (cf.\ \eqref{defwpqb}) and cannot be encoded in the tensor fields $\{\gamma,\ellc,\elltwo,\bY\}$. 

For this reason, we split the process of building notions of Killing horizons of order zero/one 
in two different levels.\ We start with a weaker definition which only restricts the metric hypersurface data and which is truly at the abstract level, in the sense that no embedding into an ambient space is required.\ In a second stage we assume the data to be embedded and add extra restrictions that  enforce  that the remaining components of the deformation tensor also vanish on the hypersurface.

Obviously, to define abstractly the notion of Killing horizon of order zero 
we need a privileged gauge-invariant vector field $\ovkil$ on the data $\hypdata$, pointing along the degeneration direction of $\gamma$.\ 
We want to allow for the possibility that $\ovkil$ has zeroes on $\N$, but we do not admit that its zero set has non-empty interior, as in such case we
would have open subsets of $\N$ where the data is not restricted at all.\ 
This assumption is natural also because Killing vectors cannot vanish on any codimension one subset of the spacetime 
unless identically zero.\ 
%
%
This leads us to the following 
definition.
\begin{definition} \label{defKHD0} (Abstract Killing horizon of order zero, AKH$_0$) Consider null hypersurface data $\{\mathcal{N},{\gamma},\bs{\ell},\ell^{(2)},\bY\}$ admitting a gauge-invariant vector field $\ovkil\in\textup{Rad}\gamma$.  
  Define $\mathcal{S}\defi \{p\in\mathcal{N}\spc\vert\spc \ovkil\vert_p=0\}$. Then $\{\mathcal{N},{\gamma},\bs{\ell},\ell^{(2)},\bY\}$ is an abstract Killing horizon of order zero if 
  $(i)$ $\mathcal{S}$ has empty interior and $(ii)$ $\pounds_{\ovkil}\gamma=0$. 
\end{definition}
Since by $(i)$ $\N\setminus\mathcal{S}$ is dense in $\N$ and
%
$\ovkil$ is proportional to $n$, condition $(ii)$ is equivalent to $\bU=0$ because (cf. \eqref{liefngammaANDU})
\begin{equation}
    \pounds_{\ovkil}\gamma=2\alpha\bU.
  \end{equation}
  As mentioned elsewhere, when the data $\hypdata$ is embedded  
  $\bU$ coincides
  with the second fundamental of $\phi(\N)$ along the null normal
  $\nu = \phi_{\star}n$.\  
  Thus, condition $(ii)$ means that $\phi(\N)$ is totally geodesic.\ In other words, \textit{an abstract Killing horizon of order zero $\hypdata$ is the 
  abstract version of a totally geodesic null hypersurface equipped with an extra
    tangent and null vector field $\eta\vert_{\phi(\N)}\defi \phi_{\star}\ovkil$
    with a mild restriction on its set of zeroes}.\ 
  For the rest of the paper, we call the vector fields $\ovkil$, $\eta\vert_{\phi(\N)}$ \textit{symmetry generators}
and the submanifolds $\mathcal{S}$, $\phi (\mathcal{S})$ \textit{fixed points sets}.

Definition \ref{defKHD0} places no restriction on the components $\bqone$ and $\p$ of the deformation tensor of $\kil$.\  So, it does not correspond to a \textit{full} Killing horizon of order zero.\ To capture the full notion we need to restrict ourselves to the embedded case.\ The definition 
is as follows. 
\begin{definition} \label{defKH0}(Killing horizon of order zero, KH$_{0}$) 
Consider an abstract Killing horizon of order zero $\hypdata$ with symmetry generator $\ovkil\in\textup{Rad}\gamma$ and assume that $\hypdata$ is $\{\phi,\rig\}$-embedded in a spacetime $(\mathcal{M},g)$.\ Then, $\phi(\N)$ is a Killing horizon of order zero 
if there exists an extension $\kil$ of $\phi_{\star}\ovkil$ to a neighbourhood $\mathcal{O}\subset\M$ of $\phi(\N)$ such that 
\begin{align}
\label{condKH0}\phi ^{\star}(g(\pounds_{\rig}\eta,\rig))=-\dfrac{1}{2}\pounds_{\ovkil}\elltwo,\qquad\phi ^{\star}(g(\pounds_{\rig}\eta,\cdot))=-\pounds_{\ovkil}\ellc\quad\text{on}\quad\N.
\end{align}
\end{definition}
Observe that, since 
$\pounds_{\rig}\kil\vert_{\phi(\N)}=(\nabla_{\rig}\kil-\nabla_{\kil}\rig)\vert_{\phi(\N)}$ 
and this 
involves no transverse derivatives of the rigging,  
there is no need to extend the rigging vector field $\rig$ off $\phi(\N)$ in Definition \ref{defKH0}.\ 

%
Let us now prove that Definition \ref{defKH0} indeed guarantees that $\Kkil=0$ on $\phi (\mathcal{N})$.  
\begin{proposition}\label{prop:Keta=0:KH0}
The deformation tensor $\Kkil$ is everywhere zero on any KH$_{0}$.
\end{proposition}
\begin{proof}
  From \eqref{pullbackdeftensornull} and $\bU=0$ we get that 
  the \textit{tangent-tangent} components of $\Kkil$ are automatically zero. Concerning $\Kkil(\rig,\rig)$, we compute
\begin{align}
\Kkil(\rig,\rig)=2g(\nabla_{\rig}\eta,\rig)=2g(\pounds_{\rig}\eta+\nabla_{\eta}\rig,\rig)=2g(\pounds_{\rig}\eta,\rig)+\eta(g(\rig,\rig)).
\end{align}
So \eqref{condKH0} implies 
$\phi ^{\star}(\Kkil(\rig,\rig))=-\pounds_{\ovkil}\elltwo+\ovkil(\elltwo)=0$ and hence $\Kkil(\rig,\rig)=0$ on $\phi(\N)$.
Now let $X$ be a vector field tangent to $\phi(\N)$.\ Combining \eqref{omegas} and \eqref{nablarig3} one obtains
\begin{equation}
\label{pullbacknablaxirig}\phi^{\star}(\nabla_{X}\bs{\rig})=\bF(X,\cdot)+\bY(X,\cdot),\qquad \bs{\rig}\defi g(\rig,\cdot),
\end{equation}
where we have used \eqref{prod4} and that $\nu$ is normal to $\phi(\N)$, hence $\phi^{\star}\bs{\nu}=0$. 
Then, from the fact that $\Kkil(\rig,X)=
g(\nabla_{\rig}\eta,X)+g(\nabla_{X}\eta,\rig)
= g(\nabla_{\rig}\eta,X)+ X ( g(\eta,\rig)) - g ( \eta, \nabla_X \rig)$
it follows 
\begin{align}
\nonumber \phi^{\star}\lp \Kkil(\rig,X)\rp=&\spc 
%
%
\phi^{\star}\lp g(\pounds_{\rig}\eta,X)+(\nabla_{\eta}\bs{\rig})(X)+X(\alpha )- (\nabla_{X}\bs{\rig})(\eta)\rp\\
\nonumber =&\spc \phi^{\star}\lp g(\pounds_{\rig}\eta,X)\rp+\alpha (\bsone(X)+\bs{r}(X))+X(\alpha )-\alpha   (\bF(X,n)+\bs{r}(X))\\
\nonumber =&\spc \phi^{\star}\lp g(\pounds_{\rig}\eta,X)\rp+2\alpha \bsone(X)+X(\alpha )= \phi ^{\star}\lp g(\pounds_{\rig}\eta,\cdot)\rp(X)+(\pounds_{\ovkil} \ellc)(X)=0,
\end{align}
where in the first line we used $g(\rig,\eta)\vert_{\phi(\N)}=\alpha g(\rig,\nu)\vert_{\phi(\N)}=\alpha$, in the second line we
inserted 
\eqref{pullbacknablaxirig} and in the last 
line 
we used that  $\bF$ is antisymmetric, \eqref{sigmathing1} 
  and Definition \ref{defKH0}. 
\end{proof}
Having defined Killing horizons of order zero  both at the abstract and the embedded levels, we proceed with the order one.\ The key object to consider is the tensor $\Sigmakil\defi \pounds_{\kil}\nabla$.\  Since we obviously want 
horizons of order one to be a subset of those of order zero, we restrict
ourselves to data $\hypdata$ with $\bU=0$.\ In such case, it follows from
\eqref{Sigma:U=0:xprssn} that, for any tangent vectors
$Z_1,Z_2\in\Gamma(T\N)$, the quantity
$\Sigmakil(\phi_{\star}Z_1,\phi_{\star}Z_2)$
can be written in terms of $\{\gamma,\ellc,\elltwo,\bY\}$ and $\alpha$
exclusively.\ On the other hand, by \eqref{calPab:expression} and \eqref{def:calP} we know that  $\big(\bs{\rig}(\Sigmakil)\big) (\phi_{\star}Z_1,\phi_{\star}Z_2)$
contains information about the zeroth and the \textit{first} order derivatives
of $\Kkil$ (because of the presence of the tensor $\bcalY$ in \eqref{calPab:expression}).\ These two ingredients can be combined to introduce a sensible definition of abstract Killing horizon of order one.\ Another reason why  $\Sigmakil$ is useful to define Killing horizons of
order one comes from identity \eqref{Sigma:Curvature}.\ Any Killing vector $\widetilde{\kil}$ in $(\M,g)$ satisfies the well-known property $0=\nabla_{\alpha} \nabla_{\beta}\widetilde{\kil}^{\mu}+R{}^{\mu}{}_{\beta\nu\alpha}\widetilde{\kil}^{\nu}$ (see e.g.\ \cite{frolov2012black}).\ 
Together with \eqref{Sigma:Curvature}, this suggest that the first order information can be codified by requiring that some components of $\Sigmakil$ vanish on the hypersurface.\ These considerations, combined with Lemma \ref{lem:Sigma:eta:general},  lead us to the following definition (recall that $\bs{r}$ is defined in \eqref{defY(n,.)andQ}).\
\begin{definition} \label{defKH1}(Abstract Killing horizon of order one, AKH$_{1}$) Consider null hypersurface data $ \{\mathcal{N},{\gamma},\bs{\ell},\ell^{(2)},\bY\}$ admitting a gauge-invariant vector field $\ovkil\in\textup{Rad}\gamma$.\   
  Let $\alpha\in\Fcal(\N)$ be given by $\ovkil=\alpha n$.\  
  Then, $\hypdata$ defines an abstract Killing horizon of order one if it is an abstract Killing horizon of order zero and, in addition,
  \begin{align}
        0= 
         \nablao_{a} \nablao_b \alpha  
+ 2 (\nablao_{(a} \alpha)
         (\sone_{b)} - r_{b)} ) +          2 \alpha  \nablao_{(a} \sone_{b)}  +
  n (\alpha) \Y_{ab} 
         - \alpha\pounds_{n}\Y_{ab}.
          \label{condAKH1}
        \end{align}  
\end{definition}
An abstract Killing horizon of order one embedded in $(\M,g)$ does not need to satisfy \eqref{condKH0}, so in general it does not define a Killing horizon of order zero.\  This may be confusing at first sight.\ The key to understand the terminology is the word ``\textit{abstract}".\ Any definition that carries ``abstract'' in the name is fully insensitive to the data being embedded or not, hence to any kind of extension of $\kil$.\  Since Definition \ref{defKH0} of a Killing horizon of order zero is embedded and requires the existence of an  extension of $\kil$ with specific properties, it follows that 
  \textit{abstract Killing horizons of order one may not correspond to Killing horizons of order zero, even if the former happen to be embedded}.

We have motivated Definition  \ref{defKH1} in terms of $\Sigmakil$.\ However, we still need to explain
in what sense Definition
\ref{defKH1} is connected to the vanishing of certain components of
$\Sigmakil$ and also to the vanishing
of the first transverse derivative of the deformation tensor $\Kkil$.\ To establish the link, we must of course assume that the data is embedded.\ We start by proving that \textit{a $\{\phi,\rig\}$-embedded AKH$_1$
necessarily satisfies
$\Sigmakil(\phi_{\star}Z_1,\phi_{\star}Z_2)=0$, $Z_1,Z_2\in\Gamma(T\N)$ and $\bcalP=0$}.\ By \eqref{calPab:expression}, condition $\bcalP=0$ can be understood as a restriction on the quantities $\{\w,\bqone,\bcalY\}$ introduced in \eqref{defwpqb}-\eqref{defwpqb:Y} (note that $\p$ plays no role, as $\bU=0$).\ It turns out that when the horizon satisfies the additional restriction $\bqone=0$ (which indeed occurs for Killing horizons of order zero, see  Proposition \ref{prop:Keta=0:KH0}), 
the pull-back of the first transverse derivative of $\Kgen$ 
vanishes everywhere on the hypersurface.\  
%
%
\begin{proposition}\label{lem:Sigma:eta:0:KH}
  Let   $\hypdata$ be an abstract Killing horizon of order one
  $\{\phi,\rig\}$-embedded in a spacetime $(\M,g)$, and 
  $\eta$ be any extension of $\phi_{\star}\ovkil$ off $\phi(\N)$.\  
  Then,
\begin{align}
\Sigmakil(\phi_{\star}Z_1,\phi_{\star}Z_2)=0,\quad\forall Z_1,Z_2\in\Gamma(T\N)\qquad\text{and}\qquad
\label{Sigma:w:qone:bcalY} \calP_{ab} 
=0.
\end{align}
If, in addition, condition $\bqone=0$ holds, then 
\begin{equation}
	\label{Sigma:w:qone:bcalY:2}
	\phi^{\star}(\pounds_{\rig}\Kkil)=0.
\end{equation}
\end{proposition}
\begin{proof}
  Since $\bU=0$, the quantity
    $\Sigmakil(\phi_{\star}Z_1,\phi_{\star}Z_2)$ is given by \eqref{Sigma:U=0:xprssn} and its vanishing is a direct consequence of \eqref{condAKH1}.\ 
    The claim \eqref{Sigma:w:qone:bcalY}
    then follows from 
    \eqref{Sigma:U=0:xprssn2}.\ Setting $\bU=0$ and $\bqone=0$ (hence $\w=\bqone(n)=0$, cf.\ \eqref{defwpqb}) in \eqref{calPab:expression} and using $\bcalP=0$ gives $\bcalY=0$, which proves \eqref{Sigma:w:qone:bcalY:2}. 
\end{proof}
Propositions \ref{prop:Keta=0:KH0} and \ref{lem:Sigma:eta:0:KH} relate the abstract Definitions \ref{defKHD0} and 
\ref{defKH1} with the deformation tensor $\Kgen$ and its first transverse 
derivative $\pounds_{\rig}\Kgen$ at the horizon.\ 
In particular, 
they 
allow us to conclude that 
%
%
\textit{an embedded abstract Killing horizon of order one which is, in addition, a Killing horizon of order zero   
satisfies the two conditions $\Kgen\vert_{\phi(\N)}=0$ and $\phi^{\star}(\pounds_{\rig}\Kgen)\vert_{\phi(\N)}=0$}.
We emphasize 
that 
the crucial fact that has allowed us to define an abstract notion of Killing horizon of order one is the property that $\Sigma(\phi_{\star}\vecY,\phi_{\star}\vecZ)$ 
depends only on hypersurface data quantities.\ In turn, this relies essentially on the  condition $\bU=0$ (cf.\ \eqref{Sigma:U=0:xprssn}).\ If the right hand side of \eqref{Sigma:U=0:xprssn} had involved any combination of $\{\bqone,\w,\p,\bcalY\}$ depending on the behaviour of $\kil$ off $\phi(\N)$, then it would have been impossible to impose conditions only on the data $\hypdata$ guaranteeing the validity of $\Sigmakil(\phi_{\star}\vecY,\phi_{\star}\vecZ)=0$ in the embedded picture.

%
%

There are additional refinements that one can do in this classification of horizons.\ For instance, in some situations it can be of interest to impose only a subset of all defining conditions.\ In fact, in Section \ref{sec:connection:NEH,MKH} 
the following notion of abstract \textit{weak} Killing horizon of order one plays a role.
  
%

\begin{definition} \label{def:weak:KH1}(Abstract weak Killing horizon of order one, AWKH$_1$)
A abstract weak Killing horizon of order one is an abstract Killing horizon of order zero that satisfies, in addition, 
\begin{align}
0=  \alpha\pounds_{n}(\bsone-\bs{r})+d\big(n(\alpha)\big)+\Q d\alpha.
\label{cond:weak:AKH1}
\end{align}  
\end{definition}
The name comes from the fact that condition \eqref{cond:weak:AKH1} is precisely the
contraction of \eqref{condAKH1} with $n^a$, so the horizon is more special than an abstract Killing horizon of order zero, but less restrictive than an abstract Killing horizon of order one. To prove this claim 
it suffices to insert
the identities (recall that $\bU=0$) 
\begin{align*}
n^a\nablao_a\nablao_b\alpha&=n^a\nablao_b
\nablao_a\alpha=\nablao_b\big( n(\alpha)\big)-\big(\nablao_a\alpha\big)\big(\nablao_b n^a\big)\stackbin{\eqref{nablaonnull}}=\nablao_b\big( n(\alpha)\big)-n(\alpha)\sone_b,\\
2n^a\big(\nablao_{(a}\alpha\big)(\sone_{b)}-r_{b)})&=n(\alpha)(\sone_b-r_b)-\big(\nablao_b\alpha\big)r_an^a
\stackbin{\eqref{defY(n,.)andQ}}=n(\alpha)(\sone_b-r_b)+\Q\nablao_b\alpha,\\
  2\alpha n^a\nablao_{(a}\sone_{b)}&=\alpha\lp \pounds_{n}\sone_b+\nablao_b(\bsone(n))-2\sone_a\big(\nablao_bn^a\big)\rp\stackbin{\eqref{nablaonnull}}=\alpha \pounds_{n}\sone_b, 
\end{align*}
and $n^a\pounds_n\Y_{ab} \stackbin{\eqref{defY(n,.)andQ}}=\pounds_nr_b$ into the contraction of \eqref{condAKH1} with $n^a$.

We conclude the section with a comment on how the defining conditions of abstract Killing horizons of order zero/one read in coordinates adapted to the hypersurface.\ 
Let $\hypdata$ be null hypersurface data $\{\phi,\rig\}$-embedded in a spacetime $(\M,g)$.\ 
Locally, 
one can always construct  
Gaussian null coordinates \cite{kunduri2013classification} $\{u,v,x^A\}$ 
in which $\phi(\N)\equiv\{u=0\}$ and the metric $g$ reads
\begin{equation}
g=-2dv\Big( du+u h_A(u,v,x^B) dx^A+u H(u,v,x^B) dv\Big)+q_{AB}(u,v,x^B)dx^Adx^B.
\end{equation}
For convenience, we introduce the notation 
\begin{align*}
\widehat{h}_A\defi h_A(0,v,x^B),\quad \widehat{H}\defi H(0,v,x^B),\quad \widehat{q}_{AB}\defi q_{AB}(0,v,x^B).
\end{align*}
With the \textit{choice} of rigging $\rig=-\cp_{u}\vert_{\phi(\N)}$,  
the 
data fields $\gamma$, $\ellc$, $\elltwo$ and $n$ are given by (cf.\ \eqref{emhd}, \eqref{prod1}-\eqref{prod2}) 
\begin{align*}
\gamma=\widehat{q}_{AB}dx^A\otimes_sdx^B,\qquad \ellc=dv,\qquad\elltwo=0,\qquad n=\cp_{v},
\end{align*}
so in particular 
$\bF=0$ and hence $\bsone=0$ (recall \eqref{threetensors}).\ The coordinate $v$ therefore defines the null direction of $\N$ and a foliation of the hypersurface by spacelike cross-sections.\ 
The quantity $\Q$ and the pull-backs of $\bs{r}$ and $\bY$ to the leaves $\{u=0,v=\text{const.}\}$ 
read (cf.\ \eqref{defY(n,.)andQ}, \eqref{YtensorEmbDef}) 
\begin{align*}
\Q=-\widehat{H},\qquad r_A=\frac{\widehat{h}_A}{2},\qquad \Y_{AB}= -\frac{1}{2}\cp_{u}q_{AB}\vert_{\phi(\N)}.
\end{align*}
In fact, a direct calculation shows that 
$\Q$ coincides with the surface gravity of 
$n=\cp_{v}$.\ 

Now, if $\hypdata$ is an AKH$_0$ with symmetry generator $\ovkil=\alpha n$ then the tensor field $\bU=\frac{1}{2}\pounds_n\gamma$ (cf.\ \eqref{threetensors}) vanishes identically.\ In coordinates, this translates into 
$\cp_v\widehat{q}_{AB}=0$, so that the leaves of the foliation are all isometric (which is a well-known property of totally geodesic null hypersurfaces).\ If, in addition, $\hypdata$ constitutes an AKH$_1$ then a straightforwawrd computation allows one to write equations  \eqref{condAKH1} as
\begin{align}
	\nn	0&=\cp_v\cp_v\alpha-\alpha \cp_v\widehat{H}-\widehat{H}\cp_v\alpha, \\
	\nn 0&=\cp_{x^A}\cp_v\alpha -\frac{\alpha}{2} \cp_v \widehat{h}_A-\widehat{H}\cp_{x^A}\alpha,\\
	\nn 0&=D_AD_B\alpha - \widehat{h}_{(A}D_{B)}\alpha-\frac{1}{2}(\cp_v\alpha)\cp_uq_{AB}\vert_{\phi(\N)}+\frac{\alpha}{2}\cp_v\cp_uq_{AB}\vert_{\phi(\N)}.
	%
	%
\end{align}
where $D$ is the Levi-Civita covariant derivative on each leaf $\{u=0,v=\text{const.}\}\subset\N$.\ 
Note that since $\alpha$ does not vanish on open sets, these equations can be seen as transport equations for $\widehat{H}$, $\widehat{h}_A$ and $\widehat{q}_{AB}$.\ Note also that the presence of zeroes of $\alpha$ affects greatly the type of equations that need to be solved.	

\section{Connection with non-expanding and isolated horizons}\label{sec:connection:NEH,MKH}

At the embedded level, a standard Killing horizon clearly encompasses any of the previously introduced notions of horizons.\ Nevertheless, there exist weaker versions of horizons in the literature, among which non-expanding, weakly isolated, and isolated horizons \cite{ashtekar2000generic,ashtekar2002geometry,ashtekar2000isolated,gourgoulhon20063+,jaramillo2009isolated,krishnan2002isolated} stand out.\ A natural question that arises at this point is 
what is the connection 
with our definitions above.\  
We devote this section to answering this question.\ 
Besides establishing the connection with previous notions, the discussion will also allow us to emphasize what are the main differences. 
As we shall see, 
our abstract notions 
have three advantages with respect to 
other sorts of horizons introduced in the literature, 
namely that $(i)$ they \textbf{do not require the existence of an ambient space} where the horizon needs to be embedded, $(ii)$ they \textbf{do not restrict the topology of the horizon} and $(iii)$ they enable the \textbf{existence of fixed points} within the horizon.\ 

The organization of this section is as follows.\ We first revisit basic results of totally geodesic null hypersurfaces and horizons (see e.g.\ \cite{duggal2007null, galloway2004null,gourgoulhon20063+} for further details) and then we connect our abstract notions with 
the concepts of non-expanding, weakly isolated and isolated horizons.  


Consider a totally geodesic (i.e.\ with vanishing second fundamental form) null hypersurface $\nullhyp$ embedded in a spacetime $(\M,g)$.\   
Given a null vector field $\kil\in\Gamma(T\nullhyp)$ with zero set $\mathcal{S}\defi \{p\in\nullhyp\spc \vert \spc \kil\vert_p=0\}$, one can define \cite{gourgoulhon20063+} a one-form $\bkilone$ on $\nullhyp\setminus\mathcal{S}$ by 
\begin{equation}
	\label{def:bkilone}\nabla_X\kil\spc\stackbin{\nullhyp\setminus\mathcal{S}}=\spc\bkilone(X)\kil\qquad \forall X\in\Gamma\big(T(\nullhyp\setminus\mathcal{S})\big).
\end{equation}
Non-expanding, weakly isolated, and isolated horizons are defined as follows \cite{ashtekar2000generic,ashtekar2002geometry,ashtekar2000isolated,gourgoulhon20063+,jaramillo2009isolated,krishnan2002isolated}.
%
\begin{definition}\label{def:NEH}
	(Non-expanding horizon) Let $\nullhyp$ be a null hypersurface embedded in a spacetime $(\M,g)$.\ Then, $\nullhyp$ is a non-expanding horizon if it is totally geodesic and has product topology $S\times\mathbb{R}$, where $S$ is a cross-section 
	of $\nullhyp$ and the null generators are along $\mathbb{R}$. 
\end{definition}
\begin{definition}\label{def:WIH}(Weakly isolated horizon) A weakly isolated horizon is a non-expanding horizon $\nullhyp$ endowed with a null generator $\kil$ such that 
	\begin{equation}
		\label{def:eqWIH}(\pounds_{\kil}\bkilone)(X)\stackbin{\nullhyp}=0, \qquad\forall X\in\Gamma(T\nullhyp),
	\end{equation}
	where $\bkilone$ is a one-form field along $\nullhyp$ satisfying \eqref{def:bkilone}.
\end{definition}
\begin{definition}\label{def:IsolatedHorizon}
	(Isolated horizon) Consider a non-expanding horizon $\nullhyp$ and let $\ovnabla$ be the connection induced\footnote{For a general embedded null hypersurface the connection induced from the ambient space depends on the choice of a rigging (see e.g.\ \cite{mars1993geometry}), but this dependence drops off when the hypersurface is totally geodesic \cite{ashtekar2002geometry}.} on $\nullhyp$ from the Levi-Civita covariant derivative $\nabla$ of $(\M,g)$.\ Then, $\nullhyp$ is an isolated horizon if it is 
	endowed with a null generator $\kil$ such that $[\pounds_{\kil},\ovnabla_a]=0$ on $\nullhyp$.
\end{definition}


We can now compare Definitions \ref{defKHD0}, \ref{defKH1} and \ref{def:weak:KH1} with those of non-expanding, weakly isolated and isolated horizons.  
\begin{nonexphor*}
\normalfont
The connection between non-expanding horizons and abstract Killing horizons of order zero follows immediately from Definition \ref{def:NEH}.\  \textit{A non-expanding horizon is an embedded AKH$_0$ $\hypdata$ with the additional condition that $\N$ has product topology}. 
\hfill $\blacksquare$
\end{nonexphor*}

\begin{weakisolhor*}
\normalfont
%
Weakly isolated horizons and Definition \ref{def:weak:KH1} are closely related, as we shall prove next.\ In order to understand their relationship, 
we first need to obtain an explicit expression for the pull-back of $\bkilone$ to a hypersurface data.\ So, let us consider null hypersurface data $\{\N,\gamma,\ellc,\elltwo,\bY\}$ embedded in a spacetime $(\M,g)$ with embedding $\phi$ and rigging $\rig$ 
so that $\phi(\N)=\nullhyp$.\  
Given 
that 
$\phi_{\star}n$ is a null generator of $\nullhyp$
we can define a function $\alpha \in \Fcal(\nullhyp)$ by 
$\kil=\alpha \phi_{\star}n$.\   
Moreover, $\nullhyp$ being totally geodesic means that 
$\bU=0$.  

The pull-back $\phi^{\star}\bkilone$ can be expressed in terms of the one-forms $\bsone$, $\bs{r}$ and $d\alpha$ by 
combining \eqref{nablaonnull}, 
\eqref{nablaXYnablao} and \eqref{def:bkilone}.\ Specifically, one gets
\begin{align*}
\bkilone(\phi_{\star}X)\kil\hspace{0.02cm}\stackbin{\nullhyp\setminus\mathcal{S}}=\hspace{0.02cm}\nabla_{\phi_{\star}X}\kil\hspace{0.02cm}\stackbin{\nullhyp\setminus\mathcal{S}}=\hspace{0.02cm}\phi_{\star}\Big( \nablao_X(\alpha n)-\alpha\bs{r}(X)n\Big)\hspace{0.02cm}\stackbin{\nullhyp\setminus\mathcal{S}}=\hspace{0.02cm} \big( d\alpha +\alpha\lp  \bsone-\bs{r}\rp\big)(X) \hspace{0.05cm}\phi_{\star}n,\quad X\in\Gamma(T\N),
\end{align*} 
which, together with $\kil=\alpha\phi_{\star}n$, yields, at those points of $\N$ where $\alpha\neq0$,  
\begin{equation}
\label{bkilone:and:alpha:s-r}\phi^{\star}\bkilone=\frac{d\alpha}{\alpha} + \bsone-\bs{r}.
\end{equation}
We now compute the pull-back $\phi^{\star}\big(\pounds_{\kil}\bkilone\big)$ at points where $\alpha\neq0$ (we use the same notation as before and define $\ovkil=\alpha n$ by $\kil\defi \phi_{\star}\ovkil$).\ A direct calculation gives
\begin{align*}
&\phi^{\star}\big(\pounds_{\kil}\bkilone\big)\stackbin{\eqref{LiePull}}=\pounds_{\ovkil}(\phi^{\star}\bkilone)=\alpha\pounds_n(\phi^{\star}\bkilone)+(\phi^{\star}\bkilone)(n)d\alpha=\alpha\pounds_n\lp\frac{d\alpha}{\alpha} + \bsone-\bs{r}\rp+\lp\frac{n(\alpha)}{\alpha} -\bs{r}(n)\rp d\alpha\\
&=\alpha\lp\frac{d\big(n(\alpha)\big)}{\alpha}-\frac{n(\alpha) d\alpha}{\alpha^2} + \pounds_n(\bsone-\bs{r})\rp+\lp\frac{n(\alpha)}{\alpha} +\Q\rp d\alpha= d\big(n(\alpha)\big)+ \alpha\pounds_n(\bsone-\bs{r}) +\Q d\alpha, 
\end{align*}
after using \eqref{bkilone:and:alpha:s-r}, $\bsone(n)=0$,  $\pounds_n(d\alpha)=d\big(n(\alpha)\big)$ and $-\bs{r}(n)=\Q$ (cf.\ \eqref{defY(n,.)andQ}).  

It is now immediate to check that
$\phi^{\star}\big(\pounds_{\kil}\bkilone\big)=0$ (or equivalently 
$(\pounds_{\kil}\bkilone)(X)=0\spc\forall X\in\Gamma(T\nullhyp)$) holds 
for any (embedded) abstract weak Killing horizon of order one (because of condition \eqref{cond:weak:AKH1}).\ Combining this fact with \eqref{def:eqWIH}, we conclude that \textit{a weakly isolated horizon is an embedded abstract weak Killing horizon of order one together with the additional restrictions that $(i)$ the horizon has product topology 
\textup{(because a weakly isolated horizon is a non-expanding horizon, see Definition \ref{def:WIH})} 
and $(ii)$ that the symmetry generator $\kil$ is everywhere non-zero on the horizon}. 
%
%
Definition \ref{def:weak:KH1} enables much greater flexibility concerning the existence of fixed points on the horizon.\ This is not a minor point since the presence of fixed points can have important implications.\ For instance, a standard Killing horizon with vanishing surface gravity can have zeroes along some generators.\ The need to erase them can have unpleasant consequences, such as turning a null hypersurface with compact cross sections into a horizon with non-compact sections.\ Our definition does not have this problem.\ This can have interesting consequences, e.g.\ in the classification of near horizon geometries. \hfill $\blacksquare$
\end{weakisolhor*}

\begin{isolhor*}
\normalfont
Consider now 
an isolated horizon $\nullhyp$ 
with generator $\kil$.\ The connection $\ovnabla$ induced on $\N$ from the ambient Levi-Civita covariant derivative $\nabla$ is given by (cf.\  \eqref{nablaXYnablao})
\begin{equation}
\label{def:ovnabla:U=0}\ovnabla_XW\hspace{0.05cm}\defi \hspace{0.05cm}\nablao_XW-\bY(X,W)n\quad\spc\forall X,W\in\Gamma(T\N),
\end{equation}
or, what is the same,
\begin{equation}
	\label{def:ovnabla:U=0:christoffel} \stackrel{\circ}{\Gamma}\hspace{-0.1cm}{}^c_{ab}-\ov{\Gamma}{}^c_{ab}=n^c\Y_{ab},
\end{equation}
where $\stackrel{\circ}{\Gamma}\hspace{-0.1cm}{}^c_{ab}$, $\ov{\Gamma}{}^c_{ab}$ are the corresponding Christoffel symbols of $\nablao$ and $\ovnabla$.\ 
Thus, condition $[\pounds_{\kil},\ovnabla_a]=0$ in Definition \ref{def:IsolatedHorizon} amounts to impose that the tensor Lie derivative of $\ovnabla$ along $\ovkil$ vanishes on the abstract manifold $\N$ (cf.\ \eqref{commutation}), namely that $\ov{\Sigma}\defi\pounds_{\ovkil}\ovnabla=0$.\ We therefore compute $\ov{\Sigma}$ 
explicitly and then check whether setting $\ov{\Sigma}=0$ indeed corresponds to any of our abstract definitions of horizons.\ As a prior step, it is convenient to prove the identity
\begin{align}
	\label{onemoreeq:1}\pounds_{\ovkil}\lp \F_{ab}+\Y_{ab} \rp=\alpha\big( \nablao_a\sone_b-\nablao_b\sone_a\big) +\alpha \pounds_{n} \Y_{ab}+(\nablao_{a}\alpha )\lp \sone_b+r_b \rp -(\nablao_{b}\alpha )\lp \sone_{a}-r_{a} \rp. 
\end{align}
By direct computation, we get
\begin{align}
	\nn \pounds_{\ovkil}\lp \F_{ab}+\Y_{ab} \rp&=\alpha \pounds_{n}\lp \F_{ab}+\Y_{ab} \rp+(\nablao_{a}\alpha )\lp \F_{cb}+\Y_{cb} \rp n^c+(\nablao_{b}\alpha )\lp \F_{ac}+\Y_{ac} \rp n^c\\
	\label{onemoreeq:2} &=\alpha \pounds_{n}\lp \F_{ab}+\Y_{ab} \rp+(\nablao_{a}\alpha )\lp \sone_b+r_b \rp -(\nablao_{b}\alpha )\lp \sone_{a}-r_{a} \rp, 
\end{align}
where in the second line we have used that $\bF$ is antisymmetric together with definitions 
\eqref{threetensors}-\eqref{defY(n,.)andQ}.\ 
Inserting  
Cartan's formula $(\pounds_n\bF)_{ab}=\big( d\bF(n,\cdot)+d(\bF(n,\cdot))\big)_{ab}\stackbin{\eqref{threetensors}}=(d\bsone)_{ab}=\nablao_a\sone_b-\nablao_b\sone_a$ into \eqref{onemoreeq:2},  equation \eqref{onemoreeq:1} follows at once.\ 
Now, 
%
%
in analogy with 
\eqref{commutation}, 
the tensor $\ov{\Sigma}$ satisfies \cite{yano1957Lie}
\begin{align}
	\pounds_{\ovkil} \ovnabla_{a } T_{b_1 \cdots b_p } = \ovnabla_{a } \pounds_{\ovkil} T_{b_1 \cdots b_p } - \sum_{\mathfrak{i}=1}^p  \ov{\Sigma}^{d}{} _{a  b_\mathfrak{i}} T_{b_1 \cdots b_{\mathfrak{i}-1} d\hspace{0.03cm} b_{\mathfrak{i}+1} \cdots b_p }\quad \textup{for any $(0,p )$ tensor $T$ on $\N$.} \label{commutation:IHs}
\end{align}
Setting $T=\gamma$  in \eqref{commutation:IHs} yields 
$0 =   \ov{\Sigma}^{d}{} _{a b} \gamma_{dc} +  \ov{\Sigma}^{d}{} _{a  c} \gamma_{b d}$
because (recall $\gamma(n,\cdot)=0$ and $\bU=0$)
\begin{align*}
	\pounds_{\ovkil}\gamma\stackbin{\eqref{liefngammaANDU}}=2\alpha\bU=0\qquad\text{and}\qquad   \ovnabla_{a} \gamma_{bc}\stackbin{\eqref{def:ovnabla:U=0:christoffel}}=\nablao_{a}\gamma_{bc}\stackbin{\eqref{nablaogamma}}=0.
\end{align*}
Performing the same cyclic permutation of indices
%
%
%
%
that 
gives  
\eqref{Sigma:eta:def:tensor}, 
one obtains 
%
%
%
%
%
%
%
%
$\gamma_{cd}\ov{\Sigma}^{d}{} _{a b}=0$, so $\ov{\Sigma}$ lies in the kernel of $\gamma$ and therefore 
$\ov{\Sigma}^{c}{} _{a b}=\ell_d\ov{\Sigma}^{d}{} _{a b} n^c$.\ We can then compute the contraction $\ell_d\ov{\Sigma}^{d}{} _{a b}$ by particularizing \eqref{commutation:IHs} for $T=\ellc$ and using \eqref{prod2}, \eqref{sigmathing1}, \eqref{def:ovnabla:U=0} and the fact that $\bsone(n)=0$.\ Specifically, one finds
\begin{align*}
  \ell_{ d}\ov{\Sigma}^{d}{} _{a  b} =&\spc \ovnabla_{a } \pounds_{\ovkil} \ell_{b} -	\pounds_{\ovkil} \ovnabla_{a } \ell_{b} 
  = \nablao_a\lp 2\alpha\sone_b+\nablao_b\alpha\rp+n(\alpha) \Y_{ab}  -	\pounds_{\ovkil}\lp \nablao_a\ell_b+\Y_{ab} \rp\\
  =&\spc  2\alpha\nablao_a \sone_b+2(\nablao_a \alpha)\sone_b+ \nablao_a\nablao_b\alpha+n(\alpha) \Y_{ab}  -	\pounds_{\ovkil}\lp \F_{ab}+\Y_{ab} \rp\\
  =&\spc  2\alpha\nablao_{(a} \sone_{b)}+2(\nablao_{(a} \alpha)(\sone_{b)}- r_{b)}) + \nablao_a\nablao_b\alpha+n(\alpha) \Y_{ab}   -\alpha \pounds_{n} \Y_{ab} ,
\end{align*}
%
%
%
where in the second and third lines we have used \eqref{nablaoll} with $\bU=0$ and \eqref{onemoreeq:1} respectively.\ 
%
%
%
In particular, $\ov{\Sigma}^{c}{}_{ab}=\ell_{ d}\ov{\Sigma}^{d}{} _{a  b} n^c =0$ for an abstract Killing horizon of order one (by \eqref{condAKH1}).\ We conclude that \textit{an isolated horizon is an embedded abstract Killing horizon of order one satisfying the additional restrictions that $(i)$ the horizon has product topology and $(ii)$ the symmetry generator $\kil$ vanishes no-where at the horizon.} \hfill $\blacksquare$
\end{isolhor*}
The connection between Definitions \ref{defKHD0}, \ref{defKH1} and \ref{def:weak:KH1} and the standard notions of non-expanding, weakly isolated and isolated horizons sets the possibility of generalizing  results obtained in the context of the latter to arbitrary topologies and to horizons with fixed points.\ 

\section{An application: fixed points of symmetry generators}\label{sec:application:zeroes}

As an application of the framework introduced in this paper we study 
the fixed points set of the symmetry generator $\ovkil$ of an abstract horizon.\ Zeroes of Killing vectors are well-understood.\ However, since the notions that we have introduced are substantially weaker, it is of interest to see what can be said about the zeroes in that case.

As a prior step, we revisit the  notions of transverse submanifold and cross-section of a null hypersurface data $\hypdata$.\ For further details we refer to \cite{manzano2023constraint}.\  
%
A \textit{transverse submanifold} is a codimension-one embedded submanifold of $\N$ to which $n$ is everywhere transverse.\ Such $S$ always exists on sufficiently local domains of $\N$.\ Given the embedding $\psi: S\longhookrightarrow\N$ of $S$ in $\N$, the tensor $h\defi\psi^{\star}\gamma$ defines a metric on $S$,  
so one can introduce its Levi-Civita covariant derivative $\nabh$.\ A special case of transverse submanifold is when $S$ 
is intersected precisely once by each generator.\ In such case we say that $S$ is a cross-section (or a ``section'' for short).\ Note that existence of a section entails global restrictions on the data, in particular on the topology of $\N$.\ 

A particularly relevant result in this context is that, when the data admits a cross-section, one can always find a gauge transformation that makes the 
scalar $\kappa_n$ equal to zero.\ We prove this next.
\begin{proposition}\label{lem:gauge:kappa}
	Consider null hypersurface data $\hypdata$ admitting a cross-section $S$
	and define  
	$\kappa_n$ by \eqref{defY(n,.)andQ}.\ Then, there always exists a choice of gauge in which $\G_{(z,V)}(\Q)=0$.\
\end{proposition}
\begin{proof}
	The proof is based on the equation 
	\begin{equation}
	\label{ODE:hatz}\pounds_n x-\Q x=0,
	\end{equation}
	which is a linear homogeneous ODE and  hence admits a global unique solution for $x$ provided 
	initial data $x\vert_{ S }$.\ We first prove that $x\vert_S\neq0$ necessarily implies $x \neq 0$.\ For that we argue by contradiction.\ Suppose that $x\vert_S\neq0$ and that $x$ vanishes  at some point $q\in\N$.\ Equation 
	\eqref{ODE:hatz} is an ODE along the generator containing $q$, so by uniqueness of solutions $x$ vanishes everywhere along the  generator.\ Since all
	generators intersect $S$, this contradicts the hypothesis that
	$x\vert_S\neq0$.\  
	Therefore, $x\vert_S\neq0$ indeed guarantees $x\neq0$.\ Now, 
	take gauge  parameters $\{z=x,V\in\Gamma(T\N)\}$, where $x$ is a solution of \eqref{ODE:hatz}
	with $x\vert_S\neq0$.\ Then, for any $V\in\Gamma(T\N)$,  $\mathcal{G}_{(z,V)}(\Q)=0$ as a consequence of the gauge transformation \eqref{Qprime}. 
	%
	%
\end{proof}
We can now study the nature of the zeroes on abstract Killing horizons.\ We start with a lemma that shows that fixed points of AKH$_1$s 
are always simple, i.e.\ 
that $d\alpha$ cannot vanish at 
fixed points.
\begin{lemma}\label{lem:alpha:dalpha:kappa:zero:AKH1}
Let $\hypdata$ be an AKH$_1$
with symmetry generator $\ovkil$.\ Define  
$\alpha,\kappa\in\Fcal(\N)$ by $\ovkil\defi \alpha n$ and \eqref{defkappaonN} respectively.\ Then, there cannot exist a point $p\in \N$ where $\{\alpha\vert_p=0,d\alpha\vert_p=0\}$. 
\end{lemma}
\begin{proof}
  We argue by contradiction. So, assume the existence of $p\in \N$ where $\alpha\vert_p=0$ and $d\alpha\vert_p=0$. Consider an autoparallel (in the connection $\nablao$) curve $\mathcal{C}
  : (a,b) \subset \mathbb{R}  \longrightarrow \N$ starting at $p$ and with affine parameter $t$,  i.e.\ a curve with  tangent vector $X$ satisfying $X^a\nablao_aX^b=0$, $X(t)=1$ and $\mathcal{C} (0) = p$. We shall use the notation $\frac{df}{dt}=X(f)$ and $\frac{d^2f}{dt^2}=X^a\nablao_a(X^b\nablao_bf)=X^aX^b\nablao_a\nablao_bf$ for any function $f\in\Fcal(\N)$. The proof is based on the defining condition
  \eqref{condAKH1} of AKH$_{1}$ 
 and its contraction with $n$, given by 
 \eqref{cond:weak:AKH1}.
 Contracting these equations with $X$ 
 yields 
\begin{align}
	\label{Lie:id:sigma:beforeXX} 0&= \frac{d^2\alpha}{dt^2} +2(\sone_b-{r}_b)X^b\frac{d\alpha}{dt}+X^aX^b\lp n(\alpha)\Y_{ab}+\alpha \big( 2\nablao_{a}\sone_{b}-(\pounds_{n}\bY)_{ab}\big)\rp ,\\ 
	\label{Lie:id:sigma:eq:for:n(alpha)} 0&=\frac{d \big(n(\alpha)\big)}{dt}+\Q \frac{d\alpha}{dt}+\alpha\pounds_{n}(\bsone-\bs{r})_aX^a.
\end{align}
This is a system of ordinary differential equations for $\alpha$ and $n(\alpha)$ admitting 
$\{\alpha=0,n(\alpha)=0\}$ as solution.\  The initial conditions are
$\{\alpha|_p=0, n(\alpha)|_p=0\}$ so by uniqueness  we conclude
$\{\alpha=0,n(\alpha)=0\}$ along  $\mathcal{C}$.\ The  initial  vector $X|_p$
of the autoparallel curve is arbitrary.\ It follows that 
$\alpha=0$ on an open neighbourhood of $p$, which contradicts condition $(i)$ in Definition \ref{defKHD0}.\ 
\end{proof}
\begin{remark}\label{rem:isolated:zeroes}
	\normalfont
  Condition $\{\alpha\vert_p,d\alpha\vert_p=0\}$ is equivalent to  
  $\{\alpha\vert_p=0,\kappa\vert_p=0,X(\alpha)\vert_p=0\spc\spc\forall X\!\in\! T_pS\}$
  on any transverse submanifold $(S,h)$ of $\N$ containing $p$.\
  This is a
  consequence of \eqref{defkappaonN} and 
  $X(\alpha)\vert_p=d\alpha(X)\vert_p=0$.\ 
  The following interesting implications are worth mentioning:
  \begin{itemize}\itemsep0cm
  \item[\textup{\textbf{1.}}] Condition $\{\alpha\vert_p=0,\kappa\vert_p\neq0\}$ is compatible with $X(\alpha)\vert_p=0$ for all $X\in T_pS$.\ This indeed occurs for bifurcate Killing horizons, where the Killing vector vanishes on a spacelike section (namely on the bifurcation surface) but is non-zero everywhere else in the horizon.
  \item[\textup{\textbf{2.}}] If $\{\alpha\vert_p=0,\kappa\vert_p=0\}$, then there must exist \textit{at least one} vector $W\in T_pS$ verifying $W(\alpha)\vert_p\neq0$.\ 
  
  In particular, 
  if $X(\alpha)\vert_p\neq0$ for \textit{all} $X\in T_p S$ then the point $p \in \N$ is isolated in 
  $(S,h)$,  
  i.e.\ there exists a neighbourhood $\mathcal{U}_p\subset S$ of $p$ where $\ovkil$ vanishes only at $p$.\ Another interesting case occurs when there exists a non-empty subspace $\Pi_p\defi\{X\in T_pS\spc\vert\spc X(\alpha)\vert_p=0\}$ of $ T_pS$.\  
  Then $\alpha$ is identically zero on the submanifold of $S$ generated by taking all possible $\nablao$-autoparallel curves 
  starting at $p$ with initial tangent vector in $\Pi_p$.\ This follows directly from \eqref{Lie:id:sigma:beforeXX}-\eqref{Lie:id:sigma:eq:for:n(alpha)}, since for any such autoparallel curve $\mathcal{C}_p$ we have $\alpha\vert_p=0$, $n(\alpha)\vert_p=0$ (by $\kappa\vert_p=0$, cf.\ \eqref{defkappaonN}) and hence $\alpha$ must necessarily vanish everywhere along $\mathcal{C}_p$.\ 
  \end{itemize}
  \end{remark}
%
%
%
%
%
%
%
%
%
%
%
We have all the ingredients to study the properties of the fixed points set of 
an abstract Killing horizon of order one.\ In fact, partial results hold also in the more general setting of an abstract weak Killing horizon of order one.\ For simplicity we state the proposition under the assumption that the data admits a cross-section.\  
Since local cross-sections always exist, this restriction is inessential and can be easily lifted (at the expense of a more cumbersome presentation).

\begin{theorem}
	\label{lemzerosets}
	Let $\hypdata$ be an abstract weak Killing horizon of order one with symmetry generator $\ovkil$ and fixed points set $\mathcal{S}$.\  
	Then, the surface gravity $\kappa$ of $\ovkil$ is constant along the null generators.\ 
	Assume that $\N$ admits a cross-section $S\subset\N$ and select the gauge so that $\Q=0$ (cf.\ Proposition \ref{lem:gauge:kappa}).\ Then 
	$\ovkil$ and $\mathcal{S}$ are given by
	\begin{equation}
		\label{final:ovkil:n} \ovkil=(f+\kappa\lambda)n,\qquad\mathcal{S}\defi\{p\in\N\spc\vert\spc f(p)+\kappa(p)\lambda(p)=0\},
	\end{equation}
	where $f,\lambda\in\Fcal(\N)$ are functions satisfying $n(\lambda)=1$, $n(f)=0$.\ Now, choose any point $p\in\N$ (not necessarily a fixed point).\ Then,
        	\begin{itemize}\itemsep0cm
		\item[$(i)$] If $\kappa\vert_p\neq0$ there exists a neighbourhood $\mathcal{U}_p$ of $p$ where $\kappa\neq0$, and either  
		$\mathcal{S}\cap\mathcal{U}_p=\emptyset$ or   
		$\mathcal{S}\cap\mathcal{U}_p=\{p\in\mathcal{U}_p\spc\vert\spc  f(p)+\kappa(p)\lambda(p)=0\}$ defines a transverse submanifold of $\N$.\		
              \item[$(ii)$] If $\kappa\vert_p=0$, consider the
    generator             $\mathcal{C}_p$ containing $p$.\ Then
     $\ovkil\vert_{\mathcal{C}_p}$ is either 
    identically zero or vanishes nowhere.\
              \item[$(iii)$] Restrict $\hypdata$ to be an abstract Killing horizon of order one.\ Then $\mathcal{S} = \mathcal{S}_1 \cup
                \mathcal{S}_2$, where  $\mathcal{S}_1$ is the (possibly empty) union of 
                generators defined by 
		the simultaneous zeroes of the functions $f$ and $\kappa$,
                and $\mathcal{S}_2$ is a (possibly empty) union of transverse submanifolds of $\N$.\ Moreover, $\kappa$ vanishes identically on ${\mathcal{S}_1}$
                and nowhere on ${\mathcal S}_2$.
              		\end{itemize}
              \end{theorem}
              \begin{remark}\normalfont
\tcb{Regarding $(iii)$, the structure of $\mathcal{S}_1$ must be compatible with  \textup{Remark \ref{rem:isolated:zeroes}}, item \textup{\textbf{2}}.}
              \end{remark}
\begin{proof}
Inserting $n(\alpha)=\kappa-\alpha\Q$ (cf.\ \eqref{defY(n,.)andQ}, \eqref{defkappaonN}) into \eqref{cond:weak:AKH1} yields 
\begin{align}
  d\kappa=\alpha\big(  d\Q-\pounds_{n}(\bsone-\bs{r})\big),
		 \label{Lie:id:sigma:times:n}
\end{align}
Contracting  with $n$ gives $n(\kappa)=0$ after using \eqref{defY(n,.)andQ} and $\bsone(n)=0$.\ Thus, $\kappa$ is constant along the generators of $\N$.\ To prove \eqref{final:ovkil:n} we first apply Proposition \ref{lem:gauge:kappa} and fix a gauge where  $\Q=0$. Then, the combination of \eqref{defY(n,.)andQ} and \eqref{defkappaonN} gives
\begin{equation}
	\label{n(alpha):equals:kappa}n(\alpha)=\kappa.
      \end{equation}
      Select any function on $\N$ satisfying $n(\lambda)=1$ (e.g.\ solve the linear ODE
      $n(\lambda)=1$ with initial data $\lambda |_S =0$).\ Then $f \defi \alpha - \kappa \lambda$ satisfies $n(f)=0$. This demonstrates \eqref{final:ovkil:n}. 

      Define $u \defi f + \kappa \lambda$.\ The fixed point set corresponds to
      the zeroes of $u$, i.e.\ $\mathcal{S} = \{ q \in \N;  u(q)=0\}$.\
      Take any point $p\in\N$.\ If  $\kappa\vert_p\neq0$ then it is also non-zero
      in a sufficiently small neighbourhood  $\mathcal{U}_p\subset\N$ of $p$.\ Then either $\mathcal{S}\cap\mathcal{U}_p=\emptyset$ or $\mathcal{S}$  defines a
            transverse submanifold (it is a hypersurface of $\N$ because $du$ is non-zero on $\cu_p$ and transverse because in fact $n (u) = \kappa \neq 0$).\ 
            If $\kappa\vert_p=0$, then $\ovkil \vert_{{\mathcal C}_p} = f|_p n$, because
            $\kappa\vert_{{\mathcal C}_p} = 0$ and $n(f)=0$. This proves items $(i)$ and $(ii)$.

            For item $(iii)$ we assume that $\hypdata$ is an abstract Killing horizon of order one.\ Define $\mathcal{S}_1 \defi \mathcal{S} \cap \{\kappa = 0\}$
              and $\mathcal{S}_2 \defi \mathcal{S} \cap \{\kappa \neq 0\}$, so that $\mathcal{S} = \mathcal{S}_1 \cup \mathcal{S}_2$. 
              By item $(i)$, $\mathcal{S}_2$ is a (possibly empty) union of transverse submanifolds.\ Assume that $\mathcal{S}_1\neq\emptyset$ and take a point $p \in \mathcal{S}_1$.\ Define $q$ to be 
              the intersection of the cross-section $S$ with the generator $\mathcal{C}_p$ containing $p$.\ Then $\kappa(q) = f(q)=0$ and $\mathcal{C}_p \subset \mathcal{S}_1$.\ 
              It follows that
              $\mathcal{S}_1$ is a (possibly empty) union of 
              generators defined by the simultaneous zeroes of the functions $f$ and $\kappa$.
\end{proof}

\section*{Acknowledgements}
The authors acknowledge financial support under the project PID2021-122938NB-I00 (Spanish Ministerio de Ciencia, Innovaci\'on y Universidades and FEDER ``A way of making Europe'') and  SA096P20 (JCyL). M. Manzano also acknowledges the Ph.D. grant FPU17/03791 (Spanish Ministerio de Ciencia, Innovaci{\'o}n y Universidades).

\begingroup
\let\itshape\upshape
\bibliographystyle{acm}

\bibliography{ref}

\end{document}